\documentclass[conference]{IEEEtran}
\usepackage[linesnumbered,ruled,vlined]{algorithm2e}%in-list
\usepackage{amsmath}%in-list%in-acmart
\usepackage{amssymb}%in-list
\usepackage{amsthm}%in-list
\usepackage{bm}%in-list
\usepackage{booktabs}%in-acmart
\usepackage[font={small}]{caption}%in-acmart
\usepackage[usenames]{color}%in-list
\usepackage{float}%in-list
\usepackage{graphicx}%in-list%in-acmart
\usepackage{graphics}%in-list
\usepackage{makecell}%in-list
\usepackage{stfloats}%in-list
\usepackage{subfig}%in-list
\usepackage{siunitx}%in-list
\usepackage{setspace}%in-list

\usepackage{epstopdf}%in-list%shall be included after graphics/graphicx

\usepackage{cite}
\usepackage[font={small}]{caption}

\newtheorem{Definition}{Definition}

\newtheorem{Lemma}{Lemma}
\newtheorem{Theorem}{Theorem}
\newtheorem{Policy}{Policy}

\newtheorem{Remark}{Remark}

\newtheorem{BM}{BM}

\usepackage{enumitem}

\usepackage{titlesec}
\titlespacing\section{0pt}{8pt plus 4pt minus 4pt}{0pt}
\titlespacing\subsection{0pt}{8pt plus 4pt minus 4pt}{0pt}

\IEEEoverridecommandlockouts

\usepackage{geometry}% FOR EDAS: left=0.62in,right=0.62in,top=0.75in,bottom=1.00in
\begin{document}

\title{
	\huge{
		Dynamic Uploading Scheduling in mmWave-Based Sensor Networks via Mobile Blocker Detection
	}
}

\author{
Yifei~Sun{\textsuperscript\textdagger}{\textsuperscript\textdaggerdbl},
~Bojie~Lv{\textsuperscript\textdagger},
~Rui~Wang{\textsuperscript\textdagger},
~Haisheng~Tan{\textsuperscript\textsection},
~and~Francis~C.~M.~Lau{\textsuperscript\textdaggerdbl}\\
{\textsuperscript\textdagger}Southern University of Science and Technology, Shenzhen, China \\
{\textsuperscript\textdaggerdbl}The University of Hong Kong, Hong Kong, China \\
{\textsuperscript\textsection}University of Science and Technology of China, Hefei, China
}

\maketitle
\begin{abstract}
	The freshness of information, measured as Age of Information (AoI), is critical for many applications in next-generation wireless sensor networks (WSNs).
	Due to its high bandwidth, millimeter wave (mmWave) communication is seen to be frequently exploited in WSNs to facilitate the deployment of bandwidth-demanding applications.
	However, the vulnerability of mmWave to user mobility typically results in link blockage and thus postponed real-time communications.
	In this paper, joint sampling and uploading scheduling in an AoI-oriented WSN working in mmWave band is considered, where a single human blocker is moving randomly and signal propagation paths may be blocked.
	The locations of signal reflectors and the real-time position of the blocker can be detected via wireless sensing technologies.
	With the knowledge of blocker motion pattern, the statistics of future wireless channels can be predicted.
	As a result, the AoI degradation arising from link blockage can be forecast and mitigated.
	Specifically, we formulate the long-term sampling, uplink transmission time and power allocation as an infinite-horizon Markov decision process (MDP) with discounted cost.
	Due to the curse of dimensionality, the optimal solution is infeasible.
	A novel low-complexity solution framework with guaranteed performance in the worst case is proposed where the forecast of link blockage is exploited in a value function approximation.
	Simulations show that compared with several heuristic benchmarks, our proposed policy, benefiting from the awareness of link blockage, can reduce average cost up to 49.6\%.
\end{abstract}

\IEEEpeerreviewmaketitle

\section{Introduction}
In wireless sensor networks (WSNs), the timeliness of status information sampled at the sensor locations (e.g., camera snapshots or broadband radar signals) is critical for effective monitoring in many real-time and data-intensive applications.
With the ever-increasing dense sensor deployment and bandwidth demand, communication in millimeter wave (mmWave) band has great potential to be integrated in WSNs due to its wide spectrum.
However, because of the sparse propagation paths and the pencil-shaped beams for overcoming the extremely high path loss, mmWave communications may suffer from the blockage problem in dynamic environments.
This will cause severe channel fluctuation and link outage, making the delivered status information outdated \cite{6515173}.

To quantify the freshness of the status information, the concept of \textit{Age of Information} (AoI) has been adopted as a metric\cite{Kenney-2011}.
AoI is defined as the time elapsed since the status update is sampled at the sensor.
From the joint sensing and transmission scheduling perspective, the AoI of samples delivered from the sensors should be as fresh as possible in order to guarantee accurate monitor and synchronous control.
There have been a number of research works on AoI scheduling design in various WSNs\cite{8123937,9729335,8778671,8938128,8736523,9796654,BoZhou-2018-GC}.
In \cite{8123937}, the authors minimized the long-term average AoI of a single sensor with an energy harvesting battery.
The scheduling design was formulated as a continuous time stochastic control problem.
In \cite{9729335}, the problem of status update control in an EH-enabled source via an mmWave link was formulated as a Markov decision process (MDP).
In addition to the single sensor scenario, the works in \cite{8778671,8938128} extended the joint sampling and transmission design to the multi-sensor scenario.
\cite{8778671} minimized the summation of average AoIs of multiple sensors under the sampling and transmission energy constraints via an MDP formulation.
This work was further extended in \cite{8938128} by considering transmission failures and non-uniform sample size.
However, the algorithms proposed in many of the existing works, such as the Q-learning algorithm used in \cite{9729335} and the policy iteration algorithm used in \cite{BoZhou-2018-GC}, are prohibitive in computation complexity which grows exponentially with respect to the number of sensors.
Although, the linear approximation of value functions adopted in \cite{8778671,8938128} can reduce the computation complexity, it is difficult to analyze their performance analytically.
As a result, how to design approximate MDP (AMDP) for AoI-oriented scheduling design with a guarantee on worst-case performance and low computational complexity is still unanswered.

Because of the small wavelength, the link blockage has become one of the open issues in mmWave communications, especially in indoor scenarios where swarms of mobile users may roam around constantly.
Fortunately, with the technique of wireless sensing, the room layout can be reconstructed in advance \cite{JADE,E-Mi}, and the location of the human body can be detected in real time \cite{10005216}, such that potential link degradation due to blockage can be predicted.
To our best knowledge, there exists no previous work on the AoI-oriented scheduling design problem exploiting the environment and human motion detection.
In fact, all the existing works, e.g., \cite{9729335,8778671,8938128,9796654,BoZhou-2018-GC}, adopted the stochastic channel modes, which are oblivious to the signal propagation environment \cite{9593198}.
Hence, link blockage due to human mobility can hardly be captured.

In order to exploit wireless sensing in AoI-oriented scheduling design, we consider joint sampling and uploading scheduling in an mmWave-based sensor network with awareness of the communication environment which includes the static signal reflectors and one dynamic human blocker.
Specifically, in the mmWave-based WSN, the sensors detect their targets, save the sensing samples in their buffers, and deliver the data to a server via an mmWave uplink.
A human blocker is moving randomly, which blocks some signal paths, causing a significant degradation in the corresponding uplink channel.
Sampling at the sensor and uplink transmission are jointly scheduled to minimize the AoI of the data samples at the monitor and energy consumption for sampling and uplink transmission.
The joint scheduling is formulated as an infinite-horizon MDP with discounted cost, so that knowledge of the mmWave propagation channel and human dynamics can be exploited.
The main contributions of this work are summarized as follows:
\begin{itemize}[leftmargin=10pt]
	\item
	      A predictive scheduling framework is provided for environment-aware transmission scheduling.
	      In the proposed MDP formulation, the future AoI degradation due to potential link blockage is naturally considered in the current scheduling according to the Bellman's equations.
	\item
	      We propose a low-complexity AMDP framework with a guarantee of the worst-case performance.
	      Specifically, we first introduce a decoupling principle to design heuristic scheduling polices as the reference policy, whose average cost (i.e., value function) can be derived analytically.
	      With the expression of the value function, the above policy iteration can be formulated analytically, whose optimization efficiency is significantly better than the conventional numerical search.
	      Simulations show that compared with several heuristic benchmarks, our proposed policy, benefiting from the awareness of the link blockage, can reduce the average cost with a high performance gain.
\end{itemize}

The remainder of this paper is structured as follows.
The system model is introduced in Section \ref{sec:system_model}.
In Section \ref{sec:problem_formulation}, the problem of dynamic sampling and uploading scheduling is formulated as an infinite-horizon MDP.
In Section \ref{sec:low-complexity}, a low-complexity suboptimal scheduling framework is proposed with a guarantee on the worst-case performance.
Finally, the performance of the proposed low-complexity scheduling scheme is verified by comparing with benchmarks in Section \ref{sec:simulation}, and the conclusion is drawn in Section \ref{sec:conclusion}.

\section{System Model}
\label{sec:system_model}
In this section, we first give an overview of the mmWave-based WSN including the blocker mobility model, and then define the channel model as well as the AoI model.
We use the following notation throughout this paper.
Bold lowercase $\mathbf{a}$ denotes a column vector, bold uppercase $\mathbf{A}$ denotes a matrix, non-bold letters $a$ and $A$ denote scalar values, and the letter $\mathcal{A}$ denotes a set.
$(a)^{+}$ denotes $\max(0,a)$.
$[\mathbf{A}]_{i,j}$, $\mathbf{A}^{\mathsf{T}}$, and $\mathbf{A}^{\mathsf{H}}$ denote the $(i,j)$-th element, transpose, and conjugate transpose, respectively.
$\mathcal{CN}(m,R)$ denotes complex Gaussian distribution with mean $m$ and variance $R$.
$\mathcal{U}[a,b]$ denotes a uniform distribution over the interval $[a,b]$.
$\mathbb{E}[.]$ denotes an expectation operator.
$\mathbb{I}[.]$ denotes an indicator function.
The main notations used are listed in Table \ref{tab:notations}.

\begin{table}[H]
	\footnotesize
	\centering
	\begin{tabular}{cl}
		\toprule
		\textbf{Symbol}        & \textbf{Description}                                         \\
		\midrule
		$K$/$\mathcal{K}$      & Number/Set of sensors                                        \\
		% $N_{\mathrm{R}}/N_{\mathrm{T}}$                               & Number of receive/transmit antenna elements                  \\
		% $\mathbb{L}$                                                  & Set of possible locations of the mobile blocker              \\
		%$l^{\mathrm{BS}}$/$l_{k}^{\mathrm{s}}$                        & Location index of the BS/$k$-th sensor                       \\
		% $T_{\mathrm{F}}$                                              & Frame duration                                               \\
		% $\mathbf{P}^{\mathrm{B}}$                                     & Transition matrix of the mobile blocker's location           \\
		% $M$                                                           & Number of NLoS paths from one sensor to the BS               \\
		% $\mathcal{M}$                                                 & Index set of propagation paths                               \\
		$\mathbf{H}_{t,k}$     & Channel matrix of $k$-th sensor in $t$-th frame              \\
		$l_{t}^{\mathrm{B}}$   & Location index of the mobile blocker in $t$-th frame         \\
		$Y_{t,k}$              & Baseband channel power gain of $k$-th sensor in $t$-th frame \\
		$Q_{t,k}$              & Queue length of $k$-th sensor in $t$-th frame                \\
		$A_{t,k}^{\mathrm{s}}$ & AoI at $k$-th sensor in $t$-th frame                         \\
		$A_{t,k}^{\mathrm{d}}$ & AoI for $k$-th sensor at the server in $t$-th frame          \\
		$s_{t,k}$              & Sampling action for $k$-th sensor in $t$-th frame            \\
		$\tau_{t,k}$           & Transmission time allocated to $k$-th sensor in $t$-th frame \\
		$p_{t,k}$              & Transmission power of $k$-th sensor in $t$-th frame          \\
		%$A_{\mathrm{max}}$                                            & AoI of outdated samples                                      \\
		% $\mathcal{S}_{t}$/$\mathcal{S}_{t,k}$                         & Global/Local system state in $t$-th frame                    \\
		% $\widetilde{\mathcal{S}}_{t}$/$\widetilde{\mathcal{S}}_{t,k}$ & Global/Local abstract state in $t$-th frame                  \\
		% $\mathcal{Y}_{t}$                                             & Set of baseband power gains of all sensors in $t$-th frame   \\
		% $\mathcal{Q}_{t}$                                             & Set of queue lengths of all sensors in $t$-th frame          \\
		% $\mathcal{A}_{t}^{\mathrm{s}}$/$\mathcal{A}_{t}^{\mathrm{d}}$ & Set of AoIs at the sensors/server in $t$-th frame            \\
		% $\mathbf{a}_{t}$/$\mathbf{a}_{t,k}$                           & Global action/Local action in $t$-th frame                   \\
		% $\Omega$/$\Pi$                                                & Scheduling policy/reference policy                           \\
		%$w_{\mathrm{P}}$/$w_{\mathrm{Q}}$                             & Weight for energy consumption/outdatedness penalty           \\
		%$W$/$W^{\Pi}$                                                 & Optimal/approximate value function                           \\
		\bottomrule
	\end{tabular}
	\caption{Main notations.}
	\label{tab:notations}
\end{table}

\vspace{-0.2cm}
\subsection{Network Model}
\label{subsec:network_model}
As illustrated in Fig. \ref{fig:network_model}, we consider an mmWave-based wireless monitoring system in a two-dimensional indoor space consisting of one server connected with the base station (BS) and $K$ sensors, where the sensors sense the targets (e.g., taking photos or detecting motions via radar waves), collect sensing samples, and deliver the sensing samples to the server via the BS.
The set of sensors is denoted as $\mathcal{K}\!\triangleq\!\{1,2,\ldots,K\}$.
The sensing sample collection, namely sampling, and sample uploading are scheduled to maintain the timeliness of sensing data at the server.
Both line-of-sight (LoS) and non-line-of-sight (NLoS) paths (e.g., reflection paths from the walls) exist between the sensors and the BS.
There is one person moving in the space, who may block the propagation paths from the sensors to the BS.
In this work, we consider the case of only one mobile blocker in the ambient environment which represents most practical scenarios while leave multi-blocker scenarios to future work.
With wireless sensing techniques\cite{9737357,9198891}, it is assumed that the location of the human blocker can be detected.
The analog MIMO architecture with single radio frequency (RF) chain and a half-wavelength uniform linear phased array (ULA) is adopted at both the BS and the sensors, so that both transmission and receiving beams can be aligned to the available paths.
The linear phased arrays at the BS and sensors have $N_{\mathrm{R}}$ and $N_{\mathrm{T}}$ antenna elements, respectively.
\begin{figure}[tb]
	\centering
	\includegraphics[width=0.96\linewidth]{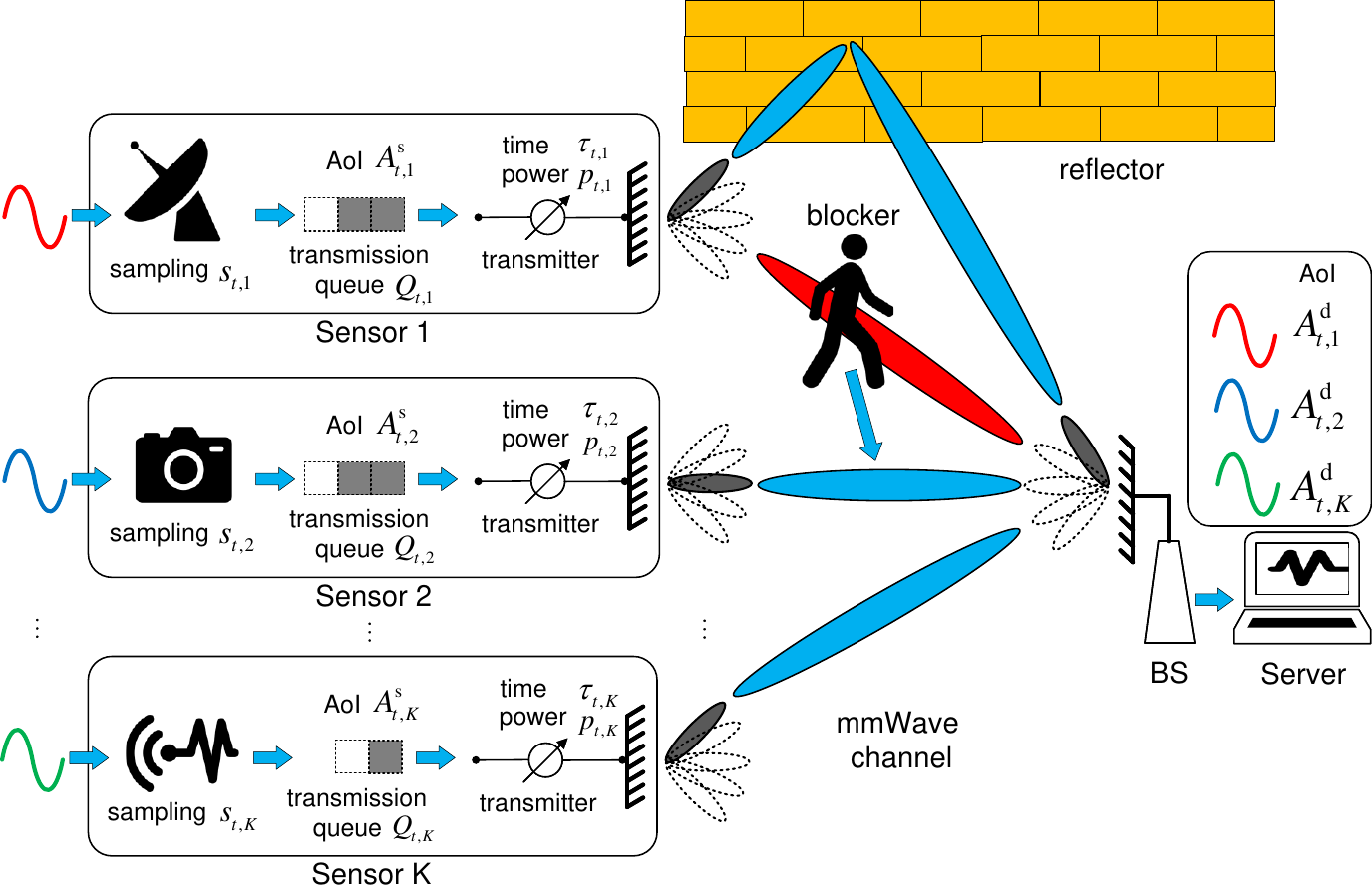}
	\caption{Network model of the considered mmWave-based WSN.}
	\label{fig:network_model}
\end{figure}

The uplink transmission time is organized by physical-layer frames with duration $T_{\mathrm{F}}$, where the channel state information (CSI) is assumed to be quasi-static in one frame.
The locations of the indoor space are quantized into grids with indexes.
Let $\mathbb{L}\triangleq\{1,2,\ldots,|\mathbb{L}|\}$ be the set of location indexes.
The location indexes of the BS and the $k$-th sensors are denoted as $l^{\mathrm{BS}}$ and $l_{k}^{\mathrm{s}}$, respectively.
The location index of the mobile human blocker in the $t$-th frame is denoted as $l_{t}^{\mathrm{B}}$.
It is assumed that the mobility of the blocker follows a time-invariant Markov chain, with the following transition probabilities,
\begin{align}
	\label{eqn:P_B}
	\Pr\big[l_{t+1}^{\mathrm{B}}\!=\!\ell^{\prime}\big|l_{t}^{\mathrm{B}}\!=\!\ell\big]=\big[\mathbf{P}^{\mathrm{B}}\big]_{\ell,\ell^{\prime}},\ \forall{t},\ \forall \ell,\ell^{\prime}\in\mathbb{L},
\end{align}
where $\mathbf{P}^{\mathrm{B}}\in\mathbb{R}^{|\mathbb{L}|\times|\mathbb{L}|}$ denotes the transition matrix of the blocker's mobility.

\subsection{Channel Model}
\label{subsec:channel_model}
In order to capture the impact of human blockage, the geometric channel model \cite{6717211} is adopted in this paper.
There are at most $M$ NLoS paths and one LoS path from one sensor to the BS.
Specifically, denote $\mathcal{M}\!=\!\{0,1,2,\ldots,M\}$ as the index set of propagation paths, where the index of the LoS path is $0$.
For the convenience of exposition, the $i$-th path of the $k$-th sensor is denoted as the $(k,i)$-th path.
Hence, the channel matrix $\mathbf{H}_{t,k}\!\in\!\mathbb{C}^{N_{\mathrm{R}}\times N_{\mathrm{T}}}$ from the $k$-th sensor to the BS in the $t$-th frame can be written as
\begin{align}
	\label{eqn:H_tk}
	\mathbf{H}_{t,k}
	=
	\sum_{i\in\mathcal{M}}B_{t,k,i}(l_{t}^{\mathrm{B}})\alpha_{t,k,i}\mathbf{a}_{\mathrm{R}}(\phi_{k,i})\mathbf{a}_{\mathrm{T}}^{\mathsf{H}}(\theta_{k,i}),
\end{align}
where $B_{t,k,i}(l_{t}^{\mathrm{B}})\!\in\!\{0,1\}$ is the blockage indicator.
Thus, $B_{t,k,i}(l_{t}^{\mathrm{B}})\!=\!0$ when $(k,i)$-th path is blocked, otherwise $B_{t,k,i}(l_{t}^{\mathrm{B}})\!=\!1$.
Moreover, $\phi_{k,i}$ and $\theta_{k,i}$ denote the angle of arrival (AoA) and angle of departure (AoD) of the $(k,i)$-th path, respectively.
$\alpha_{t,k,i}$ denotes the complex gain of the $(k,i)$-th path in the $t$-th frame obeying a complex Gaussian distribution with zero mean and variance $\rho_{k,i}^{-1}$, i.e., $\alpha_{t,k,i}\sim\mathcal{CN}(0,\rho_{k,i}^{-1})$.
The path loss $\rho_{k,i}$ depends on the path length $R$ and the reflection loss (for NLoS path).
Moreover, $\mathbf{a}_{\mathrm{R}}(\phi_{k,i})$ and $\mathbf{a}_{\mathrm{T}}(\theta_{k,i})$ represent the normalized array response vectors of the ULAs at the sensors and the BS, respectively.
Thus,
\begin{align*}
	\mathbf{a}_{\mathrm{R}}(\phi_{k,i})
	 &
	\!=\!
	\frac{1}{\sqrt{N_{\mathrm{R}}}}
	\left[1\!,e^{-j\pi\sin(\phi_{k,i})}\!,\ldots\!,e^{-j\pi(N_{\mathrm{R}}\!-\!1)\sin(\phi_{k,i})}\right]^{\mathsf{T}},
	\\
	\mathbf{a}_{\mathrm{T}}(\theta_{k,i})
	 &
	\!=\!
	\frac{1}{\sqrt{N_{\mathrm{T}}}}
	\left[1\!,e^{-j\pi\sin(\theta_{k,i})}\!,\ldots\!,e^{-j\pi(N_{\mathrm{T}}\!-\!1)\sin(\theta_{k,i})}\right]^{\mathsf{T}}.
\end{align*}

The human blocker is modeled as a disk with radius $r_{\mathrm{B}}$, and hence, the indicator $B_{t,k,i}$ can be determined by comparing the blocker radius with the shortest distances between the blocker's centroid and the $i$-th path.
Let $\mathcal{P}_{k,i}$ be the set of line segments in the $(k,i)$-th path, and $\mathbb{D}(\ell,\mathcal{P}_{k,i})$ be the minimum distance from the disk centroid to the line segments of the $(k,i)$-th path, then $
	B_{t,k,i}(l_{t}^{\mathrm{B}})
	\!=\!
	\mathbb{I}[\mathbb{D}(l_{t}^{\mathrm{B}},\mathcal{P}_{k,i})
		\!\geq\!
		r_{\mathrm{B}}]
$.
The path loss of the LoS path is usually much smaller than that of the NLoS paths.
Therefore, the uplink transmission suffers from significant degradation when the LoS path is blocked.

Let $\mathbf{w}_{t,k}\!\in\!\mathbb{C}^{N_{\mathrm{R}}\!\times\!1}$ and $\mathbf{f}_{t,k}\!\in\!\mathbb{C}^{N_{\mathrm{T}}\!\times\!1}$ be the analog combiner and precoder of the BS and the $k$-th sensor when the $k$-th sensor is transmitting in the $t$-th frame, respectively.
The uplink capacity of the $k$-th sensor in the $t$-th frame can be expressed as
\begin{align}
	\label{eqn:R_tk}
	\vspace{-0.2cm}
	R_{t,k}
	\triangleq
	W\log_{2}\Bigg(1+p_{t,k}\underbrace{\frac{\big|\mathbf{w}_{t,k}^{\mathsf{H}}\mathbf{H}_{t,k}\mathbf{f}_{t,k}\big|^{2}}{\|\mathbf{w}_{t,k}\|^{2}N_{0}W}}_{\textrm{Baseband gain $Y_{t,k}$}}\Bigg),
	\vspace{-0.2cm}
\end{align}
where $
	Y_{t,k}(l_{t}^{\mathrm{B}})
	\triangleq
	\frac{\left|\mathbf{w}_{t,k}^{\mathsf{H}}\mathbf{H}_{t,k}\mathbf{f}_{t,k}\right|^{2}}
	{\|\mathbf{w}_{t,k}\|^{2}N_{0}W}
$ denotes the baseband channel power gain, $p_{t,k}$ denotes the transmission power of the $k$-th sensor, $N_{0}$ is the noise power spectral density, $W$ is the bandwidth.
Due to the hardware constraint, the maximum transmission power satisfies
\begin{equation}
	\label{eqn:power_constraint}
	p_{t,k}\leq P_{\mathrm{max}},\ \forall{t},{k}\in\mathcal{K}.
\end{equation}

The analog precoders and combiners are chosen from a pre-defined codebook composed of a finite number of beam directions, as follows,
\begin{align}
	\mathbf{w}_{t,k}\in\mathcal{W}
	\triangleq
	\{\mathbf{a}_{\mathrm{R}}(\phi_{q}),q=1,2,\ldots,N_{\mathrm{R}}\}, \\ \mathbf{f}_{t,k}\in\mathcal{F}\triangleq\{\mathbf{a}_{\mathrm{T}}(\phi_{p}),p=1,2,\ldots,N_{\mathrm{T}}\},
\end{align}
where $
	\phi_{q}
	\!=\!
	\arcsin\Big(\frac{2(q-1)}{N_{\mathrm{R}}}\!-\!1\Big)
$, and $
	\theta_{p}
	\!=\!
	\arcsin\Big(\frac{2(p-1)}{N_{\mathrm{T}}}\!-\!1\Big)
$.
In order to avoid the costly channel matrix estimation in every frame, the transmission and receiving beams are adapted to maximize the average signal-to-noise ratio (SNR) instead of instantaneous SNR, as follows,
\begin{align}
	\label{eqn:beamforming_1}
	\left(\mathbf{w}_{t,k},\mathbf{f}_{t,k}\right)
	=
	\mathop{\arg\max}_{{}
	^{\mathbf{a}_{\mathrm{R}}(\phi_{q})\in\mathcal{W}}
	_{\mathbf{a}_{\mathrm{T}}(\phi_{p})\in\mathcal{F}}}
	\mathbb{E}_{\mathbf{H_{t,k}}}
	\big|\mathbf{a}_{\mathrm{R}}^{\mathsf{H}}(\phi_{q})\mathbf{H}_{t,k}\mathbf{a}_{\mathrm{T}}(\theta_{p})\big|^{2}.
\end{align}
Since the AoAs and AoDs of the paths in \eqref{eqn:H_tk} are static, they can be estimated before the scheduling, as in \cite{JADE,E-Mi}.
The variances of path gains $\{\alpha_{k,i},\forall{k,i}\}$, can also be measured before the transmission \cite{E-Mi}.
Given the location of the human blocker by real-time localization \cite{10005216}, the distributions of channel matrices $\{\mathbf{H}_{t,k},\forall{t,k}\}$ can be predicted.
Hence, the optimization in \eqref{eqn:beamforming_1} is feasible even if $\mathbf{H}_{t,k}$ is unknown.
Intuitively, \eqref{eqn:beamforming_1} would align the transmission and receiving beams along the path with minimum path loss.

Time-Division Multiple Access (TDMA) is adopted in each frame.
Let $\tau_{t,k}$ be the transmission time allocated to the $k$-th sensor in the $t$-th frame; the following constraints should be satisfied,
\begin{align}
	\label{eqn:total_time_constraint}
	 &
	\textstyle{
	\sum_{k\in\mathcal{K}}\tau_{t,k}=T_{\mathrm{F}},\ \forall{t},
	}
	\\
	\label{eqn:time_constraint}
	 &
	0\leq\tau_{t,k}\leq T_{\mathrm{F}},\ \forall{t,k}\in\mathcal{K}.
\end{align}
Moreover, the throughput of the $k$-th sensor in the $t$-th frame is $\tau_{t,k}R_{t,k}$.

\subsection{AoI Model for Sample Uploading}
\label{subsec:AoI_model}
If the server delivers a sensing decision to the $k$-th sensor at the beginning of the $t$-th frame, a sensing sample representing the latest target status is generated by the sensor.
Since the sensing decision is usually short, we ignore the downlink signalling overhead.
It is assumed that the data volume of each sample generated by the $k$-th sensor consists of $L_{k}$ packets, each with $N_{\mathrm{b}}$ information bits.
Let $Q_{t,k}$ be the length of the uplink transmission queue (number of uplink packets) of the $k$-th sensor in the $t$-th frame, where the sensing sample is buffered.
When new data sample is generated at one sensor, the existing packets from the previous sampling in its uplink queue will be dropped, and new packets will be added.
The sampling action will induce a constant sampling energy cost which is denoted by $C^{\mathrm{s}}$.
Let $s_{t,k}\!\in\!\{0,1\}$ be the sampling action for the $k$-th sensor in the $t$-th frame, where $s_{t,k}\!=\!0$ indicates that the $k$-th sensor will continue transmitting data packets of its current sample in the $t$-th frame, and $s_{t,k}\!=\!1$ indicates that the $k$-th sensor will start transmitting the packets of the new sample.
Thus,
\begin{align}
	Q_{t+1,k}=
	\begin{cases}
		(Q_{t,k}-D_{t,k})^{+}, & s_{t,k}=0 \\
		L_{k}-D_{t,k},         & s_{t,k}=1 \\
	\end{cases}
\end{align}
where $D_{t,k}$ denotes the departure packet number of the $k$-th sensor in the $t$-th frame.
Thus, $
	D_{t,k}
	=
	\big\lfloor\frac{R_{t,k}\tau_{t,k}}{N_{\mathrm{b}}}\big\rfloor
$.

To characterize the freshness of samples, the AoI at the $k$-th sensor, denoted by $A_{t,k}^{\mathrm{s}}$, is defined as the number of frames since the generation time of the latest sample.
Moreover, since samples with extremely large AoIs will bring no benefit to the server due to its outdatedness, we set an AoI threshold $A_{\mathrm{max}}$ to characterize outdated samples.
The AoI dynamics at the $k$-th sensor is given by
\begin{align}
	A_{t+1,k}^{\mathrm{s}}=
	\begin{cases}
		\min\{A_{t,k}^{\mathrm{s}}+1,A_{\mathrm{max}}\}, & s_{t,k}=0 \\
		1,                                               & s_{t,k}=1
	\end{cases}
\end{align}

Besides the AoIs at the sensors, the server records AoIs of the latest uploaded samples of all sensors.
Once the transmission of all packets at the $k$-th sensor is accomplished, the AoI at the server for the $k$-th sensor will be updated by the AoI at the $k$-th sensor in next frame.
Hence, let $A_{t,k}^{\mathrm{d}}$ be the AoI for the $k$-th sensor at the server in the $t$-th frame, we have
\begin{align}
	A_{t+1,k}^{\mathrm{d}}
	\!=\!
	\begin{cases}
		\min\{A_{t,k}^{\mathrm{s}}\!+\!1,A_{\mathrm{max}}\}\!,
		 &
		(Q_{t,k}\!-\!D_{t,k})^{+}\!=\!0 \\
		\min\{A_{t,k}^{\mathrm{d}}\!+\!1,A_{\mathrm{max}}\}\!,
		 &
		\mathrm{otherwise}
	\end{cases}
\end{align}

\section{Problem Formulation}
\label{sec:problem_formulation}
The average AoI at the server depends on the sampling action, time and uplink power allocations of all the frames.
Clearly, due to the random motion of the human blocker and random channel matrices, it is impossible for the BS to determine the sampling action, time and uplink power allocations of all the frames in a deterministic manner.
Instead, we shall formulate their optimization as an infinite-horizon MDP.
The system state, scheduling policy, and system cost are first defined below.
\begin{Definition}[System State]
	At the beginning of the $t$-th frame, the global system state is uniquely specified by a tuple $\mathcal{S}_{t}\!\triangleq\!(l_{t}^{\mathrm{B}},\mathcal{Y}_{t},\mathcal{Q}_{t},\mathcal{A}_{t}^{\mathrm{s}},\mathcal{A}_{t}^{\mathrm{d}})$, where $\mathcal{Y}_{t}\!\triangleq\!\{Y_{t,k}\}_{k\in\mathcal{K}}$ is the set of baseband channel power gains of all sensors, $\mathcal{Q}_{t}\!\triangleq\!\{Q_{t,k}\}_{k\in\mathcal{K}}$ is the set of remaining numbers of packets at the uplink queues, $\mathcal{A}_{t}^{\mathrm{s}}\!\triangleq\!\{A_{t,k}^{\mathrm{s}}\}_{k\in\mathcal{K}}$ is the set of AoIs at the sensors, and $\mathcal{A}_{t}^{\mathrm{d}}\!\triangleq\!\{A_{t,k}^{\mathrm{d}}\}_{k\in\mathcal{K}}$ is the set of AoIs at the server.
	Moreover, the local system state of the $k$-th sensor in the $t$-th frame is defined by $\mathcal{S}_{t,k}\!\triangleq\!(l_{t}^{\mathrm{B}},Y_{t,k},Q_{t,k},A_{t,k}^{\mathrm{s}},A_{t,k}^{\mathrm{d}})$.
\end{Definition}

\begin{Definition}[Action and Policy]
	The local scheduling action of the $k$-th sensor is defined as $\mathbf{a}_{t,k}\!\triangleq\!(s_{t,k},\tau_{t,k},p_{t,k})$, including the sampling decision, and uplink transmission time and power.
	The global scheduling action is defined as the aggregation of the local actions of all the sensors; thus, $\mathbf{a}_{t}\triangleq\{\mathbf{a}_{t,k}\}_{k\in\mathcal{K}}$.
	Hence, the scheduling policy, denoted as $\Omega$, is a mapping from the system state $\mathcal{S}_{t}$ to the scheduling actions, i.e., $\Omega(\mathcal{S}_{t})=\mathbf{a}_{t}$.
\end{Definition}

Given the scheduling policy $\Omega$, the state transition probability can be written as in \eqref{eqn:Pr}.

\begin{figure*}
	\begin{align}
		\label{eqn:Pr}
		 & \Pr\big[\mathcal{S}_{t\!+\!1}\big|\mathcal{S}_{t},\Omega(\mathcal{S}_{t})\big]
		\!=\!
		\Pr\big[l_{t\!+\!1}^{\mathrm{B}}\big|l_{t}^{\mathrm{B}}\big]\Pr\big[\mathcal{Y}_{t\!+\!1}\big|l_{t\!+\!1}^{\mathrm{B}}\big]
		\times
		\textstyle{\prod_{k\in\mathcal{K}}}\big\{(1\!-\!s_{t\!,k})\mathbb{I}\big[A_{t\!+\!1\!,k}^{\mathrm{s}}\!=\!\min\{A_{t\!,k}^{\mathrm{s}}\!+\!1,A_{\mathrm{max}}\}\big]\!+\!s_{t\!,k}\mathbb{I}\big[A_{t\!+\!1\!,k}\!=\!1\big]\big\}
		\nonumber
		\\
		\!\times\!
		 & \textstyle{\prod_{k\in\mathcal{K}}}\big\{(1\!-\!s_{t\!,k})\mathbb{I}\big[Q_{t\!+\!1\!,k}\!=\!(Q_{t\!,k}\!-\!D_{t\!,k})^{+}\big|Y_{t\!,k},\tau_{t\!,k},p_{t\!,k}\big]
		\!+\!s_{t\!,k}\mathbb{I}\big[Q_{t\!+\!1\!,k}\!=\!L_{k}\!-\!D_{t\!,k}\big|Y_{t\!,k},\tau_{t\!,k},p_{t\!,k}\big]\big\} \nonumber
		\\
		\!\times\!
		 & \textstyle{\prod_{k\in\mathcal{K}}}\big\{\mathbb{I}\big[(Q_{t\!,k}\!-\!D_{t\!,k})^{+}\!=\!0\big]\mathbb{I}\big[A_{t\!+\!1\!,k}^{\mathrm{d}}\!=\!\min\{A_{t\!,k}^{\mathrm{s}}\!+\!1,A_{\mathrm{max}}\}\big]
		\!+\!\mathbb{I}\big[(Q_{t\!,k}\!-\!D_{t\!,k})^{+}\!\neq\!0\big]\mathbb{I}\big[A_{t\!+\!1\!,k}^{\mathrm{d}}\!=\!\min\{A_{t\!,k}^{\mathrm{d}}\!+\!1,A_{\mathrm{max}}\}\big]\big\}.
	\end{align}
	\hrulefill
	\vspace{-0.5cm}
\end{figure*}

In the $t$-th frame, the per-frame cost is defined as the weighted sum of AoIs at the server, energy consumption for sampling and uplink transmission, and outdated AoI penalties at the server for all sensors.
That is,
\begin{align}
	\label{eqn:cost_function}
	g(\mathcal{S}_{t},\Omega(\mathcal{S}_{t}))
	=
	\sum\nolimits_{k\in\mathcal{K}}\big[A_{t,k}^{\mathrm{d}}
	 & +w_{\mathrm{P}}(s_{t,k}C^{\mathrm{s}}+\tau_{t,k}p_{t,k}) \nonumber     \\
	 & +w_{\mathrm{Q}}\mathbb{I}[A_{t,k}^{\mathrm{d}}=A_{\mathrm{max}}]\big],
\end{align}
where $w_{\mathrm{P}}$ and $w_{\mathrm{Q}}$ denote the weights for energy consumption and AoI outdatedness penalty, respectively.
Note that the first and third terms of \eqref{eqn:cost_function} imply a nonlinear cost function for AoI \cite{8764465}.

Hence, the overall cost from the $1$-st frame is defined as
\begin{align}
	\overline{G}(\mathcal{S}_{1},\Omega)
	\triangleq
	\lim_{T\rightarrow\infty}\left[\mathbb{E}_{\mathcal{Y},\mathcal{L}^{\mathrm{B}}}\sum_{t=1}^{T}\gamma^{t\!-\!1}g_{t}(\mathcal{S}_{t},\Omega(\mathcal{S}_{t}))\Big|\mathcal{S}_{1}\right],
\end{align}
where $\gamma\!\in\!\left(0,1\right)$ is the discount factor, and the expectation is taken on $\mathcal{Y}\!\triangleq\!\{\mathcal{Y}_{1}\!,\mathcal{Y}_{2}\!,\ldots\!,\mathcal{Y}_{T}\}$ and $\mathcal{L}^{\mathrm{B}}\!\triangleq\!\{l_{1}^{\mathrm{B}}\!,l_{2}^{\mathrm{B}}\!,\ldots\!,l_{T}^{\mathrm{B}}\}$.
As a result, the joint sampling and uploading optimization can be formulated by the following infinite-horizon MDP with discounted cost.
\begin{align*}
	\mathsf{P1}:\
	\Omega^{\star}
	=
	\textstyle{
		\mathop{\arg\min}_{\Omega}\
	}
	 & \overline{G}(\mathcal{S}_{1},\Omega) \nonumber                                                                        \\
	\mathrm{s.t.}\
	 & \textrm{Constraints in \eqref{eqn:power_constraint}, \eqref{eqn:total_time_constraint}, \eqref{eqn:time_constraint}}.
\end{align*}
The Bellman's equations of the above MDP are
\begin{align}
	\label{eqn:bellman}
	W(\mathcal{S})
	=
	 &
	\textstyle{
	\min_{\Omega(\mathcal{S})}\Big[g\big(\mathcal{S},\Omega(\mathcal{S})\big)
	}
	\nonumber
	\\
	 &
	+
	\textstyle{
		\gamma\sum_{\mathcal{S}^{\prime}}W(\mathcal{S}^{\prime})\Pr[\mathcal{S}^{\prime}|\mathcal{S},\Omega(\mathcal{S})]\Big],\
	}
	\forall\mathcal{S},
\end{align}
where $W(\cdot)$ is the value function of the optimal scheduling policy (i.e., the optimal value function), and $\mathcal{S}^{\prime}$ is the system state in the next frame given system state $\mathcal{S}$ and scheduling action $\Omega(\mathcal{S})$.
Moreover, the policy minimizing the right-hand-side (RHS) of the above Bellman's equations is proved to be the optimal one \cite{bertsekas2012dynamic}.

Note that the baseband gain in the system state is continuously and independently distributed in all frames, which can be eliminated from the value function to reduce the complexity.
Hence, we first define the local and global \textit{abstract state}\cite{li2006towards} with the baseband gain eliminated, i.e., $
	\widetilde{\mathcal{S}}_{t,k}\!\triangleq\!(l_{t}^{\mathrm{B}},Q_{t,k},A_{t,k}^{\mathrm{s}},A_{t,k}^{\mathrm{d}})
$ and $
	\widetilde{\mathcal{S}}_{t}\!\triangleq\!(l_{t}^{\mathrm{B}},\mathcal{Q}_{t},\mathcal{A}_{t}^{\mathrm{s}},\mathcal{A}_{t}^{\mathrm{d}})
$.
By taking expectation on both sides of \eqref{eqn:bellman}, the Bellman's equations with respect to the abstract state can be simplified as
\begin{align}
	\label{eqn:refined_bellman}
	W(\widetilde{\mathcal{S}})
	=
	 &
	\textstyle{
	\mathbb{E}_{\mathcal{Y}}\min_{\Omega(\mathcal{S})}\Big[g\big(\mathcal{S},\Omega(\mathcal{S})\big)
	}
	\nonumber
	\\
	 &
	+
	\textstyle{
		\gamma\sum_{\widetilde{\mathcal{S}}^{\prime}}W(\widetilde{\mathcal{S}}^{\prime})\Pr\big[\widetilde{\mathcal{S}}^{\prime}|\mathcal{S},\Omega(\mathcal{S})\big]\Big],
	}
	\ \forall\widetilde{\mathcal{S}},
\end{align}
where $\widetilde{\mathcal{S}}^{\prime}$ is the abstract state in the next frame.
The optimal value function with respect to the abstract state can be represented as follows,
\begin{align}
	\label{eqn:value_function_definition}
	 &
	W(\widetilde{\mathcal{S}})
	=
	\mathbb{E}_{\mathcal{Y}}W(\mathcal{S})
	=
	\min_{\Omega}\mathbb{E}_{\mathcal{Y},\mathcal{L}^{\mathrm{B}}}^{\Omega}
	\nonumber
	\\
	 &
	\lim_{T\rightarrow\infty}\sum\nolimits_{t=1}^{T}\gamma^{t-1}\Big[g\big(\widetilde{\mathcal{S}}_{t},\Omega(\widetilde{\mathcal{S}}_{t})\big)\Big|\widetilde{\mathcal{S}}_{1}=\widetilde{\mathcal{S}}\Big],
	\ \forall\widetilde{\mathcal{S}}.
\end{align}

With the value function of the optimal policy for an arbitrary abstract state $\{W(\widetilde{\mathcal{S}})|\forall\widetilde{\mathcal{S}}\}$ and the current system state $\mathcal{S}_{t}$, the optimal action in the $t$-th frame is given by
\begin{align}
	\label{eqn:one-step_iteration}
	\Omega^{\star}
	 & (\mathcal{S}_{t})
	=
	\mathop{\arg\min}_{\Omega(\mathcal{S}_{t})}\Big[g_{t}\Big(\mathcal{S}_{t},\Omega(\mathcal{S}_{t})\Big)
		\nonumber
	\\
	 &
		+
		\gamma\sum\nolimits_{\widetilde{\mathcal{S}}_{t+1}}W(\widetilde{\mathcal{S}}_{t+1})\Pr[\widetilde{\mathcal{S}}_{t+1}|\mathcal{S}_{t},\Omega(\mathcal{S}_{t})]\Big],
	\ \forall\mathcal{S}_{t}.
\end{align}

Although conventional approaches such as policy and value iteration can be used to find the optimal scheduling policy \cite{bertsekas2012dynamic}, they suffer from the \textit{curse of dimensionality}: due to the huge system state space, the evaluation of $\{W(\widetilde{\mathcal{S}}_{t})|\forall \widetilde{\mathcal{S}}_{t}\}$ is prohibitive.
In the following section, we shall propose a low-complexity scheduling scheme to address this issue.

\section{Low-Complexity Scheduling}
\label{sec:low-complexity}

In this section, a low-complexity scheduling scheme with analytical performance bound is proposed.
Here is a sketch of the proposed scheme:
\begin{itemize}[leftmargin=10pt]
	\item \textbf{Section \ref{subsec:reference_policy}:} We first design a heuristic scheduling policy as the reference policy, whose value function (average discounted cost) can be analytically expressed with mobile blocker detection.
	\item \textbf{Section \ref{subsec:one-step_iteration}:} Approximating the optimal value function via the above derived value function, the scheduling action of each frame can be obtained by solving the optimization problem in \eqref{eqn:one-step_iteration}.
	\item \textbf{Section \ref{subsec:alternative}:} We decouple the problem and then propose an alternative optimization algorithm to derive the suboptimal solution.
	\item \textbf{Section \ref{subsec:performance_bound}:} We analyze the performance bound and time complexity of the proposed suboptimal solution.
\end{itemize}

\subsection{Decoupled Reference Policy}
\label{subsec:reference_policy}
A heuristic reference policy is proposed to provide an expression for an achievable average discounted cost.
In fact, given the policy, the system evolves as a Markov chain, whose cost can be derived via the transition matrix.
However, the dimension of the transition matrix grows exponentially with respect to the number of sensors, which makes the above approach infeasible.
In order to address this issue, in the reference policy, the transmission time allocation of the $K$ sensors is fixed.
Hence, the state transition of the sensors can be decoupled into $K$ Markov chains with significantly smaller state space.
We adopt the following decoupled reference policy (which we refer to as reference policy); however, the design of reference policy is not unique.
\begin{Policy}[Decoupled Reference Policy $\Pi$]
	The reference policy, denoted as $\Pi\!\triangleq\!(s_{t,k}^{\Pi},\tau_{t,k}^{\Pi},p_{t,k}^{\Pi})$, is elaborated below.
	\begin{itemize}[leftmargin=10pt]
		\item
		      The sampling decision is made when the uplink queue is empty, i.e., $s_{t,k}^{\Pi}\!=\!\mathbb{I}[Q_{t,k}\!=\!0]$.
		\item
		      The transmission time allocation is proportional to the data volume of corresponding sampling, i.e., $\tau_{t,k}^{\Pi}\!=\!T_{\mathrm{F}}L_{k}/\sum_{k^{\prime}\in\mathcal{K}}L_{k^{\prime}}$.
		\item
		      The transmission power is constant, i.e., $p_{t,k}^{\Pi}\!=\!P^{\Pi}$.
	\end{itemize}
\end{Policy}

To derive the average discounted cost for each sensor, we first express the transition matrices of the local abstract state of the sensors.
Specifically, let $
	\mathbf{s}_{t\!,k}\!\in\!\mathcal{R}^{|\mathbb{L}|(L_{k}\!+\!1)A_{\mathrm{max}}^{2}\!\times\!1
	}$ be the vector representing the probabilities of all local abstract states of the $k$-th sensor, where the $\kappa(l_{t}^{\mathrm{B}},\epsilon(Q_{t\!,k},A_{t\!,k}^{\mathrm{s}},A_{t\!,k}^{\mathrm{d}}))$-th entry denotes the probability of the local abstract state $\widetilde{\mathcal{S}}_{t\!,k}\!=\!(l_{t}^{\mathrm{B}},Q_{t\!,k},A_{t\!,k}^{\mathrm{s}},A_{t\!,k}^{\mathrm{d}})$, and $
	\kappa(l_{t}^{\mathrm{B}},\epsilon(Q_{t\!,k},A_{t\!,k}^{\mathrm{s}},A_{t\!,k}^{\mathrm{d}}))
	\!\triangleq\!
	(l_{t}^{\mathrm{B}}\!-\!1)(L_{k}\!+\!1)A_{\mathrm{max}}^{2}\!+\!\epsilon(Q_{t\!,k},A_{t\!,k}^{\mathrm{s}},A_{t\!,k}^{\mathrm{d}})
$ and $
	\epsilon(Q_{t\!,k},A_{t\!,k}^{\mathrm{s}},A_{t\!,k}^{\mathrm{d}})
	\!\triangleq\!
	Q_{t\!,k}A_{\mathrm{max}}^{2}\!+\!(A_{t\!,k}^{\mathrm{s}}\!-\!1)A_{\mathrm{max}}\!+\!A_{t\!,k}^{\mathrm{d}}
$ are indexes.
Let $
	\mathbf{P}_{k}
	\!\in\!
	\mathbb{R}^{|\mathbb{L}|(L_{k}\!+\!1)A_{\mathrm{max}}^{2}\!\times\!|\mathbb{L}|(L_{k}\!+\!1)A_{\mathrm{max}}^{2}}
$ be the transition probability matrix of the local abstract state of the $k$-th sensor, we have $\mathbf{s}_{t+1\!,k}\!=\!\mathbf{P}_{k}^{\mathsf{T}}\mathbf{s}_{t\!,k}$.
In order to derive the expression of $\mathbf{P}_{k}$, we first introduce the following lemma on the distribution of departure packet number.
\begin{Lemma}
	\label{lem:PMF_D}
	With sufficiently large $N_{\mathrm{R}}$ and $N_{\mathrm{T}}$, given the reference policy $\Pi$, the precoder and combiner in \eqref{eqn:beamforming_1}, and the blocker's location $l_{t}^{\mathrm{B}}$, the probability mass function (PMF) of departure packet number can be written by \eqref{eqn:PMF_D}, where $i^{\star}\!=\!\mathop{\arg\max}_{i}B_{t,k,i}(l_{t}^{\mathrm{B}})\rho_{k,i}^{-1}$.
\end{Lemma}

\begin{proof}
	Please refer to Appendix A. %\ref{proof:PMF_D}.
\end{proof}

\addtolength{\textheight}{-0.1in}
As a result, we have the following lemma on the transition matrix $\mathbf{P}_{k}$.
\begin{Lemma}
	\label{lem:P_k}
	With sufficiently large $N_{\mathrm{R}}$ and $N_{\mathrm{T}}$, given the reference policy $\Pi$, the precoder and combiner in \eqref{eqn:beamforming_1}, and the blocker's location $l_{t}^{\mathrm{B}}$, the transition probability matrix of local abstract state of the $k$-th sensor is given by \eqref{eqn:P_k}, where $\mathbf{M}_{k}^{(\ell)}\in\mathbb{R}^{(L_{k}+1)A_{\mathrm{max}}^{2}\times(L_{k}+1)A_{\mathrm{max}}^{2}}$ is given by Table \ref{tab:M_kl}, and $\mathbf{P}^{\mathrm{B}}$ is defined in \eqref{eqn:P_B}.
\end{Lemma}

\begin{proof}
	Please refer to Appendix B. %\ref{proof:P_k}.
\end{proof}

\begin{figure*}
	\begin{align}
		\label{eqn:PMF_D}
		\Pr\left[D_{t,k}^{\Pi}=d\big|l_{t}^{\mathrm{B}}\right]
		=\exp\left(-\frac{\rho_{k,i^{\star}}N_{0}W}{P^{\Pi}}\Big(2^{\frac{dN_{\mathrm{b}}}{W\tau_{t,k}^{\Pi}}}-1\Big)\right)
		-\exp\left(-\frac{\rho_{k,i^{\star}}N_{0}W}{P^{\Pi}}\Big(2^{\frac{(d+1)N_{\mathrm{b}}}{W\tau_{t,k}^{\Pi}}}-1\Big)\right).
	\end{align}
	\begin{align}
		\label{eqn:P_k}
		\mathbf{P}_{k}=
		\begin{pmatrix}
			\left[
				\mathbf
			P^{\mathrm{B}}\right]_{1,1}\mathbf{M}_{k}^{(1)}                       &
			\left[
				\mathbf
			P^{\mathrm{B}}\right]_{1,2}\mathbf{M}_{k}^{(1)}                       & \cdots & \left[
				\mathbf
			P^{\mathrm{B}}\right]_{1,|\mathbb{L}|}\mathbf{M}_{k}^{(1)}                                       \\
			\left[
				\mathbf
			P^{\mathrm{B}}\right]_{2,1}\mathbf{M}_{k}^{(2)}                       &
			\left[
				\mathbf
			P^{\mathrm{B}}\right]_{2,2}\mathbf{M}_{k}^{(2)}                       & \cdots & \left[
				\mathbf
				P^{\mathrm{B}}\right]_{2,|\mathbb{L}|}\mathbf{M}_{k}^{(2)}
			\\
			\vdots                                                                & \vdots & \ddots & \vdots \\
			\left[
				\mathbf
			P^{\mathrm{B}}\right]_{|\mathbb{L}|,1}\mathbf{M}_{k}^{(|\mathbb{L}|)} &
			\left[
				\mathbf
			P^{\mathrm{B}}\right]_{|\mathbb{L}|,2}\mathbf{M}_{k}^{(|\mathbb{L}|)} & \cdots & \left[
				\mathbf
				P^{\mathrm{B}}\right]_{|\mathbb{L}|,|\mathbb{L}|}\mathbf{M}_{k}^{(|\mathbb{L}|)}
		\end{pmatrix}
		,
	\end{align}
	\hrulefill
\end{figure*}

\begin{table*}\tiny
	\centering
	\renewcommand{\arraystretch}{1.5}
	\begin{tabular}{|c|c|c|c|c|c|c|}
		\hline
		$Q_{t,k}$              & $A_{t,k}^{\mathrm{s}}$                                & $A_{t,k}^{\mathrm{d}}$                                &
		$Q_{t+1,k}$            & $A_{t+1,k}^{\mathrm{s}}$                              & $A_{t+1,k}^{\mathrm{d}}$                              &
		$\left[\mathbf{M}_{k}^{(\ell)}\right]_{\epsilon(Q_{t\!,k}\!,A_{t\!,k}^{\mathrm{s}}\!,A_{t\!,k}^{\mathrm{d}}),\epsilon(Q_{t\!+\!1\!,k}\!,A_{t\!+\!1\!,k}^{\mathrm{s}}\!,A_{t\!+\!1\!,k}^{\mathrm{d}})}$ \\
		\hline
		$0$                    & $1,\!\ldots\!,A_{\mathrm{max}}$                       & $1,\!\ldots\!,A_{\mathrm{max}}$                       &
		$0$                    & $1$                                                   & $1$                                                   &
		$\Pr\left[D_{t,k}^{\Pi}\!\geq\!L_{k}\Big|l_{t}^{\mathrm{B}}\!=\!\ell\right]$                                                                                                                           \\
		\hline
		$0$                    & $1,\!\ldots\!,A_{\mathrm{max}}$                       & $1,\!\ldots\!,A_{\mathrm{max}}$                       &
		$1,\!\ldots\!,L_{k}$   & $1$                                                   & $\min\{A_{t,k}^{\mathrm{d}}{+}1\!,A_{\mathrm{max}}\}$ &
		$\Pr\left[D_{t,k}^{\Pi}\!=\!L_{k}\!-\!Q_{t+1,k}\Big|l_{t}^{\mathrm{B}}\!=\!\ell\right]$                                                                                                                \\
		\hline
		$1,\!\ldots\!,L_{k}$   & $1,\!\ldots\!,A_{\mathrm{max}}$                       & $1,\!\ldots\!,A_{\mathrm{max}}$                       &
		$0$                    & $\min\{A_{t,k}^{\mathrm{s}}{+}1\!,A_{\mathrm{max}}\}$ & $\min\{A_{t,k}^{\mathrm{s}}{+}1\!,A_{\mathrm{max}}\}$ &
		$\Pr\left[D_{t,k}^{\Pi}\!\geq\!Q_{t,k}\Big|l_{t}^{\mathrm{B}}\!=\!\ell\right]$                                                                                                                         \\
		\hline
		$1,\!\ldots\!,L_{k}$   & $1,\!\ldots\!,A_{\mathrm{max}}$                       & $1,\!\ldots\!,A_{\mathrm{max}}$                       &
		$1,\!\ldots\!,Q_{t,k}$ & $\min\{A_{t,k}^{\mathrm{s}}{+}1\!,A_{\mathrm{max}}\}$ & $\min\{A_{t,k}^{\mathrm{d}}{+}1\!,A_{\mathrm{max}}\}$ &
		$\Pr\left[D_{t,k}^{\Pi}\!=\!Q_{t,k}\!-\!Q_{t+1,k}\Big|l_{t}^{\mathrm{B}}\!=\!\ell\right]$                                                                                                              \\
		\hline
	\end{tabular}
	\caption{Non-zero entries of matrix $\mathbf{M}_{k}^{(\ell)}$.}
	\label{tab:M_kl}
	\vspace{-0.5cm}
\end{table*}
Finally, the value function of the reference policy, namely, the approximate value function, is given by the following theorem.
\begin{Theorem}[Value Function of Reference Policy $\Pi$]\label{thm:analytical_expression}
	With the reference policy $\Pi$, the value function is given by
	\begin{align}
		\label{eqn:local_value_function}
		W^{\Pi}(\widetilde{\mathcal{S}}_{t})
		= & \sum_{k\in\mathcal{K}}\underbrace{(\mathbf{e}_{\kappa(l_{t}^{\mathrm{B}},\epsilon(Q_{t\!,k},A_{t\!,k}^{\mathrm{s}},A_{t\!,k}^{\mathrm{d}}))}^{|\mathbb{L}|(L_{k}+1)A_{\mathrm{max}}^{2}})^{\mathsf{T}}[\mathbf{I}-\gamma\mathbf{P}_{k}]^{-1}\mathbf{g}_{k}}_{\triangleq W_{k}^{\Pi}(\widetilde{\mathcal{S}}_{t,k})}
		\nonumber                                                                                                                                                                                                                                                                                                               \\
		  & +\frac{1}{1-\gamma}w_{\mathrm{P}}T_{\mathrm{F}}P^{\Pi},
	\end{align}
	where $\mathbf{e}_{n}^{N}$ denotes an $N\times 1$ column vector whose $n$-th element is $1$ and otherwise $0$, and $\mathbf{g}_{k}\in\mathbb{R}^{|\mathbb{L}|(L_{k}+1)A_{\mathrm{max}}^{2}\times 1}$ is the cost vector for all abstract local state at the $k$-th sensor.
	Specifically, $
		\left[\mathbf{g}_{k}\right]_{\kappa(l_{t}^{\mathrm{B}},\epsilon(Q_{t\!,k},A_{t\!,k}^{\mathrm{s}},A_{t\!,k}^{\mathrm{d}}))}
		=
		A_{t,k}^{\mathrm{d}}\!+\!\mathbb{I}[Q_{t,k}\!=\!0]w_{\mathrm{P}}C^{\mathrm{s}}\!+\!\mathbb{I}[A_{t,k}^{\mathrm{d}}\!=\!A_{\mathrm{max}}]w_{\mathrm{Q}}.
	$
\end{Theorem}

\begin{proof}
	Please refer to Appendix C.%\ref{proof:analytical_expression}.
\end{proof}

\subsection{Scheduling with Approximate Value Function}
\label{subsec:one-step_iteration}
Substituting the optimal value function $W(\widetilde{\mathcal{S}}_{t+1})$ of the problem in \eqref{eqn:one-step_iteration} with the approximate value function $W^{\Pi}(\widetilde{\mathcal{S}}_{t+1})$, the proposed scheduling policy in one frame (say the $t$-th frame) given the global system state $\mathcal{S}_{t}$ can be obtained from the following optimization problem.
\begin{equation}
	\begin{aligned}
		 & \mathsf{P2}:\
		\left\{(s_{t,k}^{\star},\tau_{t,k}^{\star},p_{t,k}^{\star})\right\}_{k\in\mathcal{K}}
		=\nonumber       \\
		 &
		\begin{aligned}
			 & \mathop{\arg\min}_{\left(s_{t,k},\tau_{t,k},p_{t,k}\right)_{k\in\mathcal{K}}}\quad
			\sum_{k\in\mathcal{K}}w_{\mathrm{P}}(\tau_{t,k}p_{t,k}+s_{t,k}C^{\mathrm{s}}) \nonumber                                                                                                                            \\
			 & \qquad+\gamma\sum_{\widetilde{\mathcal{S}}_{t+1}}W^{\Pi}(\widetilde{\mathcal{S}}_{t+1})\Pr\Big[\widetilde{\mathcal{S}}_{t+1}\Big|\mathcal{S}_{t},(s_{t,k},\tau_{t,k},p_{t,k})_{k\in\mathcal{K}}\Big], \nonumber \\
			 &
			\qquad\qquad\mathrm{s.t.}\quad\textrm{Constraints in \eqref{eqn:power_constraint}, \eqref{eqn:total_time_constraint}, \eqref{eqn:time_constraint}}.
		\end{aligned}
	\end{aligned}
\end{equation}

\begin{Remark}[Predictive Scheduling]
	Note that the LoS path is usually much better than the NLoS paths.
	If the LoS path of one sensor (say the $k$-th sensor) is likely blocked soon, the penalty of AoI outdatedness will start to incur. %be easily triggered.
	Thus, its local value function $W_{k}^{\Pi}(\widetilde{\mathcal{S}}_{t,k})$ could be large for large values of AoI $A_{t\!,k}^{\mathrm{d}}$.
	Hence, the problem $\mathsf{P2}$ tends to schedule more transmission resources to this sensor, so that its AoI can be maintained at a low level.
\end{Remark}

\subsection{Alternative Optimization Algorithm}
\label{subsec:alternative}
Problem $\mathsf{P2}$ is a mixed continuous and discrete optimization problem with coupled variables $\{s_{t,k}\}_{k\in\mathcal{K}}$, $\{\tau_{t,k}\}_{k\in\mathcal{K}}$ and $\{p_{t,k}\}_{k\in\mathcal{K}}$.
An alternative optimization algorithm is proposed in Algorithm \ref{alg:alternative} to obtain the suboptimal solution.

\IncMargin{1em}
\begin{algorithm}
	\DontPrintSemicolon
	\SetInd{0.2em}{0.7em}
	\SetKwInOut{Input}{Input}
	\SetKwInOut{Output}{Output}
	\SetKwProg{Parfor}{for}{ do in parallel}{end}

	\Input{\;
	$\mathcal{S}_{t}$: System state in the $t$-th frame\;
	$\{W_{k}^{\Pi}(\widetilde{\mathcal{S}}_{t,k})\}_{\forall k\!,\widetilde{\mathcal{S}}_{t,k}}$: Approximate local value functions via Theorem \ref{thm:analytical_expression}
	}
	\Output{\;
		$\Psi^{\infty}(\mathcal{S}_{t})$: Converged action in the $t$-th frame
	}
	\BlankLine

	$n\leftarrow 0,s_{t,k}^{(0)}\leftarrow s_{t,k}^{\Pi},\tau_{t,k}^{(0)}\leftarrow\tau_{t,k}^{\Pi},d_{t,k}^{(0)}\leftarrow d_{t,k}^{\Pi}$\label{line:init}\;
	\While{not converge}{
	$n\leftarrow n+1$\;
	\Parfor{$k\in\mathcal{K}$}{
		\textit{Solve $s_{t,k}^{(n)}$ in $\mathsf{P2.1}(n,k)$ via comparison.}\;\label{line:P21}
	}
	\Parfor{$k\in\mathcal{K}$}{
		\textit{Solve $d_{t,k}^{(n)}$ in $\mathsf{P2.2}(n,k)$ via exhaustive search.}\;\label{line:P22}
	}
	\textit{Solve $\{\tau_{t,k}^{(n)}\}_{k\in\mathcal{K}}$ in $\mathsf{P2.3}(n)$ via Lemma \ref{lem:tau_tk}.}\;\label{line:P23}
	$p_{t,k}^{(n)}=\frac{1}{Y_{t,k}}\Big[2^{\wedge}\Big(\frac{d_{t,k}^{(n)}N_{\mathrm{b}}}{\tau_{t,k}^{(n)}W}\Big)-1\Big], \forall{k}$\;
	}
	\label{line:p_tkn}
	\Return{
	$\Psi^{\infty}(\mathcal{S}_{t})\leftarrow\{(s_{t,k}^{\infty},\tau_{t,k}^{\infty},p_{t,k}^{\infty})\}_{k\in\mathcal{K}}$
	}
	\label{line:return}
	\caption{Alternative optimization algorithm.}
	\label{alg:alternative}
\end{algorithm}
\DecMargin{1em}

Instead of optimizing $\{s_{t,k}\}_{k\in\mathcal{K}}$, $\{\tau_{t,k}\}_{k\in\mathcal{K}}$ and $\{p_{t,k}\}_{k\in\mathcal{K}}$, we optimize $\{s_{t,k}\}_{k\in\mathcal{K}}$, $\{\tau_{t,k}\}_{k\in\mathcal{K}}$ and $\{d_{t,k}\}_{k\in\mathcal{K}}$ alternatively, where $d_{t,k}$ denotes the number of transmission packets from the $k$-th sensor in the $t$-th frame.
Let $\{s_{t,k}^{(n)}\}_{k\in\mathcal{K}}$, $\{\tau_{t,k}^{(n)}\}_{k\in\mathcal{K}}$ and $\{d_{t,k}^{(n)}\}_{k\in\mathcal{K}}$ be the correspondingly optimized variables in the $n$-th iteration, respectively.
Initializing the actions with the reference policy $\Pi$ (Line \ref{line:init}), i.e., $
	(s_{t,k}^{(0)},\tau_{t,k}^{(0)},d_{t,k}^{(0)})\!=\!(s_{t,k}^{\Pi},\tau_{t,k}^{\Pi},d_{t,k}^{\Pi})
$, the entire procedure of solving $\mathsf{P2}$ consists of a number of iterations, and the $n$-th iteration includes the sub-problems $\{\mathsf{P2.1}(n,k)|\forall{k}\}$, $\{\mathsf{P2.2}(n,k)|\forall{k}\}$, and $\mathsf{P2.3}(n)$ in \eqref{eqn:P2.1}, \eqref{eqn:P2.2} and \eqref{eqn:P2.3}, respectively (Line 4--8), where
\begin{equation}
	p_{t,k}(d_{t,k})
	\!=\!
	\frac{1}{Y_{t,k}}
	[2^{\frac{d_{t,k}N_{\mathrm{b}}}{\tau_{t,k}^{(n\!-\!1)}W}}\!\!\!-\!1],\
	p_{t,k}(\tau_{t,k})
	\!=\!
	\frac{1}{Y_{t,k}}
	[2^{\frac{d_{t\!,k}^{(n)}N_{\mathrm{b}}}{\tau_{t,k}W}}\!\!\!-\!1].\nonumber
\end{equation}
Notice that the originally coupled sampling decision $s_{t,k}^{(n)}$ and transmission packet number allocation $d_{t,k}^{(n)}$ of all sensors are decoupled given the transmission time allocation $
	\{\!\tau_{t,k}^{(n)}
	\!
	\}_{k\in\mathcal{K}}
$, thus the complexity is significantly reduced.
\begin{figure*}
	\begin{align}
		 &
		\begin{aligned}
			\label{eqn:P2.1}
			\mathsf{P2.1}(n,k):\
			s_{t,k}^{(n)}=
			\mathop{\arg\min}_{s_{t,k}\in\{0,1\}}\quad & w_{\mathrm{P}}s_{t,k}C^{\mathrm{s}}\!+\!\gamma\sum_{\widetilde{\mathcal{S}}_{t+1,k}}W_{k}^{\Pi}(\widetilde{\mathcal{S}}_{t+1,k})\Pr\left[\widetilde{\mathcal{S}}_{t+1,k}\Big|\mathcal{S}_{t,k},(s_{t,k},\tau_{t,k}^{(n\!-\!1)},p_{t,k}^{(n\!-\!1)})\right]
		\end{aligned} \\
		 &
		\begin{aligned}
			\label{eqn:P2.2}
			\mathsf{P2.2}(n,k):\
			d_{t,k}^{(n)}= \!\!\!\!\!
			\!
			\mathop{\arg\min}_{
			d_{t,k}:\ p_{t,k}(d_{t,k})\leq P_{\mathrm{max}}
			} \!\!\!\! \!	w_{\mathrm{P}}\tau_{t,k}^{(n\!-\!1)}p_{t,k}(d_{t,k})\!+\!\gamma
			\!\!\!
			\sum_{\widetilde{\mathcal{S}}_{t+1,k}}
			\!\!\!
			W_{k}^{\Pi}(\widetilde{\mathcal{S}}_{t+1,k})
			\!
			\Pr
			\!
			\left[\widetilde{\mathcal{S}}_{t+1,k}\Big|\mathcal{S}_{t,k},\left(s_{t,k}^{(n)},\tau_{t,k}^{(n\!-\!1)},p_{t,k}(d_{t,k}) \!\right)\!\right]
		\end{aligned}                                                                                                                                                                                \\
		 &
		\begin{aligned}
			\label{eqn:P2.3}
			\mathsf{P2.3}(n):\
			\left\{\!\tau_{t,k}^{(n)}
			\!
			\right\}_{k\in\mathcal{K}}\! \!=
			\mathop{\arg\min}_{\left\{\tau_{t,k}\right\}_{k\in\mathcal{K}}}
			\sum_{k\in\mathcal{K}}
			\!
			\tau_{t,k}p_{t,k}(\tau_{t,k}) \quad
			\mathrm{s.t.}\
			\sum_{k\in\mathcal{K}}\tau_{t,k}
			\!
			=T_{\mathrm{F}},\ 0\leq\tau_{t,k}\leq T_{\mathrm{F}},\  p_{t,k}(\tau_{t,k})\leq P_{\mathrm{max}},\ \forall{k}\in\mathcal{K},
		\end{aligned}
	\end{align}
	\hrulefill
	\vspace{-0.5cm}
\end{figure*}
The optimal solution of $\{\mathsf{P2.1}(n,k)|\forall{k}\}$ can be derived by evaluating the binary local sampling action for $s_{t,k}\!=\!0$ and $s_{t,k}\!=\!1$ and choosing the one with the smaller value of the objective function (Line \ref{line:P21}).
Similarly, the optimal solution of $\{\mathsf{P2.2}(n,k)|\forall{k}\}$ can be solved by one-dimensional search over $d_{t,k}\!\in\!\{0,1,\!\ldots\!,Q_{t,k}\}$ (Line \ref{line:P22}).

Moreover, $\mathsf{P2.3}(n)$ is a convex optimization problem and the optimal solution can be derived by the following lemma (Line \ref{line:P23}).
\begin{Lemma}
	\label{lem:tau_tk}
	The optimal solution of $\mathsf{P2.3} (n)$ is given by
	\begin{align}
		\tau_{t,k}^{(n)}
		\!=\!
		\frac{d_{t,k}N_{b}}{W} \!\max\!\Bigg\{
		\!\!\frac{\ln{2}}{1\!+\!\mathbb{W}_{0}\big(\frac{Y_{t\!,k}\nu^{\star}\!-\!1}{e}\big)}, \frac{1}{\log_{2}(1\!+\!P_{\mathrm{max}}Y_{t\!,k})} \!\!\Bigg\}, \nonumber
	\end{align}
	where $\mathbb{W}_{0}(\cdot)$ denotes the principal branch of the Lambert W function, and $\nu^{\star}$ denotes the optimal Lagrangian multiplier for equality constraint \eqref{eqn:total_time_constraint} which can be solved by
	\begin{align}
		\label{eqn:Lagragian}
		\sum_{k\in\mathcal{K}}\max\Bigg\{
		 & \frac{d_{t,k}N_{\mathrm{b}}\ln{2}}{W\big[1\!+\!\mathbb{W}_{0}\big(\frac{Y_{t,k}\nu^{\star}\!-\!1}{e}\big)\big]}, \nonumber \\
		 & \frac{d_{t,k}N_{\mathrm{b}}}{W\log_{2}(1\!+\!P_{\mathrm{max}}Y_{t,k})}\Bigg\}
		\!-\!T_{\mathrm{F}}\!=\!0.
	\end{align}
\end{Lemma}
\begin{proof}
	Please refer to the Theorem 1 in \cite{9145118} for the proof.
\end{proof}
Because the left-hand-side (LHS) of \eqref{eqn:Lagragian} is non-increasing, the bisection method can be used to solve $\nu^{\star}$.
Given $d_{t,k}^{(n)}$ and $\tau_{t,k}^{(n)}$, the transmission power allocation in the $n$-th iteration, denoted by $p_{t,k}^{(n)}$, can be derived analytically (Line \ref{line:p_tkn}).

Since the iterative optimization of the sub-problems $\{\mathsf{P2.1}(n,k)|\forall{k}\}$, $\{\mathsf{P2.2}(n,k)|\forall{k}\}$, and $\mathsf{P2.3}(n)$ would lead to a smaller or equal value of the objective function in $\mathsf{P2}$, the above iteration will always converge (Line \ref{line:return}).

\subsection{Performance Bound and Time Complexity}
\label{subsec:performance_bound}
Denote $
	\Psi^{(\!n\!)}
	: \widetilde{\mathcal{S}}_{t} \to
	\big\{s_{t,k}^{(\!n\!)},\tau_{t,k}^{(\!n\!)},p_{t,k}^{(\!n\!)} \big\}_{ k\in \mathcal{K}}
$ as the scheduling policy obtained after $n$ iterations, $
	\widetilde{W}^{\Psi^{(\!n\!)}}(\widetilde{\mathcal{S}}_{t})
$ as the corresponding value functions, and $
	\Psi^{\infty}:
	\widetilde{\mathcal{S}}_{t} \to
	\big\{s_{t,k}^{\infty},\tau_{t,k}^{\infty},p_{t,k}^{\infty}\big\}_{ k\in \mathcal{K}}
$ as the scheduling policy after convergence.
The performance of $\Psi^{\infty}$ and $\Psi^{(\!n\!)}$ can be guaranteed by the following lemma.
\begin{Lemma}[Performance Bound]\label{lem:performance_bound}
	The average discounted cost of policies $\Psi^{(\!n\!)}$ and $\Psi^{\infty}$ can be bounded by
	\begin{align*}
		W(\!\widetilde{\mathcal{S}}_{t}\!)
		\!
		\leq
		\!
		W^{\!\Psi^{\!\infty}}
		\!\!
		(\!\widetilde{\mathcal{S}}_{t}\!)
		\!\leq\! \ldots \!
		\leq\!
		\widetilde{W}^{\!\Psi^{(\!n\!)}}
		\!\!(\!\widetilde{\mathcal{S}}_{t}\!)\!
		\leq
		\!
		\ldots
		\!
		\leq
		\!
		\widetilde{W}^{\!\Psi^{(\!1\!)}}
		\!\!(\!\widetilde{\mathcal{S}}_{t}\!)
		\!\leq\!
		W^{\!\Pi}\!(\!\widetilde{\mathcal{S}}_{t}\!).
	\end{align*}
\end{Lemma}

\begin{proof}
	Since $W(\widetilde{\mathcal{S}}_{t})$ is the optimal value function, it is the lower bound of value function of an arbitrary policy.
	The proof of $
		W^{\Psi^{\!\infty}}
		\!\!
		(\widetilde{\mathcal{S}}_{t})
		\!\leq\! \ldots \!
		\leq\!
		\widetilde{W}^{\Psi^{(\!n\!)}}
		\!\!(\widetilde{\mathcal{S}}_{t})\!
		\leq
		\!
		\ldots
		\!
		\leq
		\!
		\widetilde{W}^{\Psi^{(\!1\!)}}
		\!\!(\widetilde{\mathcal{S}}_{t})
	$ and $
		\widetilde{W}^{\Psi^{(\!1\!)}}
		\!\!(\widetilde{\mathcal{S}}_{t})\!\leq\!W^{\Pi}(\widetilde{\mathcal{S}}_{t})
	$ resembles the proof of the policy improvement property in Chapter II of \cite{bertsekas2012dynamic}.
\end{proof}

For conventional value iteration method, the transmission time and power shall be first discretized into $\tau_{\mathrm{D}}$ and $p_{\mathrm{D}}$ levels, respectively.
Then the cardinalities of the abstract state space and action space can be denoted by $
	|\widetilde{\mathbb{S}}|
	\!\triangleq\!
	|\mathbb{L}|A_{\mathrm{max}}^{2K}\prod_{k\in\mathcal{K}}L_{k}
$ and $
	|\mathbb{A}|
	\!\triangleq\!
	(2\tau_{\mathrm{D}}p_{\mathrm{D}})^{K}
$, respectively.
The time complexity of the conventional value iteration is $
	O(|\widetilde{\mathbb{S}}|^{2}|\mathbb{A}|)
$ for each iteration.
On the other hand, the proposed scheme consists of two stages, namely the evaluation of the approximate value function and the alternative optimization algorithm.
The time complexity of the approximate value function is $
	O(|\widetilde{\mathbb{S}}|[|\mathbb{L}|A_{\mathrm{max}}^{2}\sum_{k\in\mathcal{K}}L_{k}]^{3})
$ and the time complexity of the alternative optimization algorithm is $
	O(|\widetilde{\mathbb{S}}|\sum_{k\in\mathcal{K}}L_{k})
$ for each iteration.
Benefiting from the analytical expressed approximate value function, a number of iterations from the initial value function to the value function of a roughly-good policy is prevented, and thus the time complexity is essentially reduced.

\addtolength{\textheight}{0.1in}
\section{Simulations and Discussions}
\label{sec:simulation}
In this section, the performance of the proposed algorithm is demonstrated via simulations, where a number of benchmarks are used in the comparison.
We summarize the key findings of our simulations as follows:
\begin{itemize}[leftmargin=10pt]
	\item Our proposed scheme can converge after only a few iterations and reduce the average per-frame cost by $13.5\%$--$49.6\%$ compared with the benchmarks.
	\item Our proposed scheme can proactively keep the AoIs of sensors with future channel degradation at a low level in example traces, which intuitively verifies the benefits of exploiting mobile blocker detection in AoI-oriented scheduling design.
	\item Our proposed scheme shows the better performance compared to the benchmarks with robustness against the number of sensors.
\end{itemize}

As illustrated in Fig. \ref{fig:layout}, we consider a 20m$\times$20m square room with walls serving as the reflectors of the NLoS paths ($M\!=\!4$), where the BS is deployed at the center block of the room and the locations of sensors are uniformly distributed near the walls.
The mobility of the human blocker with radius $r_{\mathrm{B}}\!=\!0.3\textrm{m}$ follows a modified random walk.
The probability of staying in the same grid in next frame is $0.90$, while the probabilities of moving to one of the feasible neighboring blocks in the blue region are equal.
An illustrative blocker trajectory is shown by arrows in Fig. \ref{fig:layout}.
Other parameters are summarized in Table \ref{tab:parameter}.
We evaluate the proposed algorithms with different numbers of iterations in solving $\mathsf{P2}$, namely, $\Psi^{(\!1\!)}$, $\Psi^{(\!2\!)}$, $\Psi^{(\!3\!)}$, $\Psi^{\infty}$, and compare them with the following three benchmark policies, which are referred to as BM1, BM2 and BM3, respectively.
All the benchmarks adopt the same sampling policy and transmission power allocation as the reference policy, i.e., $s_{t,k}\!=\!\mathbb{I}[Q_{t,k}\!=\!0]$ and $p_{t,k}\!=\!P^{\Pi}$, while the transmission time allocation is elaborated as follows.

\begin{figure}[tb]
	\centering
	\includegraphics[width=0.6\linewidth]{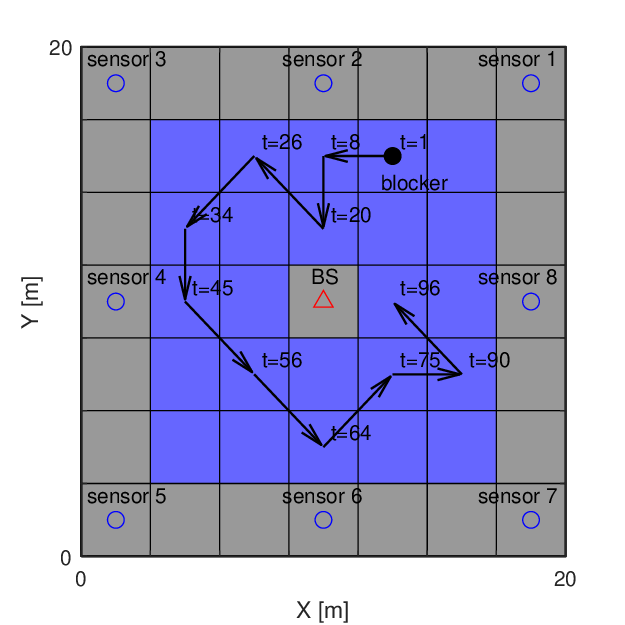}
	\caption{Simulation scenario and an illustrative trajectory of the blocker.}
	\label{fig:layout}
\end{figure}

\begin{table}[b]
	\footnotesize
	\setlength\tabcolsep{1pt}
	\centering
	\begin{tabular}{ccc}
		\toprule
		\textbf{Parameter}                      & \textbf{Symbol}                   & \textbf{Value}     \\
		\midrule
		Carrier frequency                       & $f_{\mathrm{c}}$                  & 60 GHz             \\
		\# of sensors                           & $K$                               & 4,5,6,7,8          \\
		\# of receive/transmit antenna elements & $N_{\mathrm{R}},\!N_{\mathrm{T}}$ & 64, 128            \\
		Path loss                               & $\rho_{k,i}$                      & \makecell{
		LoS: $32.5\!+\!20\log(f_{\mathrm{c}})$                                                           \\
		$\qquad\!+\!20\log(R)$\cite{maltsev2010channel}                                                  \\
		NLoS: $32.5\!+\!20\log(f_{\mathrm{c}})$                                                          \\
			$\qquad\!+\!20\log(R)\!+\!15$
		}                                                                                                \\
		\# of possible locations of the blocker & $|\mathbb{L}|$                    & 24                 \\
		Bandwidth                               & $W$                               & 400 MHz            \\
		Frame duration                          & $T_{\mathrm{F}}$                  & 10 ms              \\
		Packet size                             & $N_{\mathrm{b}}$                  & 200 KB             \\
		Noise power spectral density            & $N_{0}$                           & -174 dBm/Hz        \\
		Data volume                             & $L_{k}$                           & $\mathcal{U}(3,5)$ \\
		Threshold for outdated AoI              & $A_{\mathrm{max}}$                & 10                 \\
		Maximum transmission power              & $P_{\mathrm{max}}$                & 100 mW             \\
		Discount factor                         & $\gamma$                          & 0.98               \\
		\makecell{
		Weights for energy consumption                                                                   \\
			and outdated AoI penalty
		}                                       & $w_{\mathrm{P}},\!w_{\mathrm{Q}}$ & 10000, 100         \\
		Sampling energy consumption             & $C^{\mathrm{s}}$                  & $10^{-4}$ J        \\
		Transmission power of reference policy  & $P^{\mathrm{\Pi}}$                & 50 mW              \\
		\bottomrule
	\end{tabular}
	\caption{Parameter configuration of the simulation.}
	\label{tab:parameter}
\end{table}

\begin{BM}[Reference Policy]\label{Benchmark1}
	The transmission time allocation of each sensor is proportional to the corresponding data volume of sample at each sensor, i.e., $
		\tau_{t,k}\!=\!T_{\mathrm{F}}L_{k}/\sum_{k^{\prime}\!\in\!\mathcal{K}}L_{k^{\prime}}
	$.
\end{BM}

\begin{BM}[Largest-AoI First]\label{Benchmark2}
	The sensor with the largest AoI at the server, i.e., $\mathop{\arg\max}_{k}A_{t,k}^{\mathrm{d}}$, is scheduled for transmission sequentially, until transmission time of the frame is used up.
\end{BM}

\begin{BM}[Dynamic Backpressure \cite{georgiadis2006resource}]\label{Benchmark3}
	The sensor with the largest product of buffer length and uplink capacity, i.e., $\mathop{\arg\max}_{k}Q_{t,k}R_{t,k}$,  is scheduled for transmission sequentially, until transmission time of the frame is used up.
\end{BM}

BM1 corresponds to the reference policy we adopt in the proposed scheme with fair allocation based on uploading workload. BM2 selects UE greedily to mitigate the outdatedness AoI penalty. BM3 adopts a UE selection strategy based on both queue length and uplink channel capacity.

Fig. \ref{fig:CDF_and_barplot}(a) displays the cumulative distribution functions (CDFs) of per-frame cost of the proposed scheduling scheme as well as those of the three benchmarks, when the number of sensors is $K=8$.
The step-shape curves are caused by the outdated AoI penalty (whose weight is $w_{\mathrm{Q}}\!=\!100$ in the simulation).
It can be observed that the proposed policy with $n\!=\!1$, denoted as $\Psi^{(\!1\!)}$, has a significantly better curve than the reference policy (BM1).
Moreover, the curves of $\Psi^{(3)}$ and $\Psi^{\infty}$ are close, demonstrating that only a small number of iterations are needed in the optimization of $\mathsf{P2}$.
BM2 allocates all the transmission time to the sensors with the largest AoIs at the server, which leads to the highest probability of  per-frame costs below $150$ (one sensor with outdated AoI penalty in average).
However, the largest-AoI-first scheduling in BM2 may be stuck to the sensor whose LoS path is blocked, leading to more outdated AoI penalty.
It can be observed that compared with BM3, BM2 has much lower probability when the per-frame cost is below $250$.
The average per-frame costs are $129.3$, $256.3$, $203.5$ and $149.5$ for the proposed scheme $\Psi^{\infty}$ and the benchmarks respectively, indicating that the average performance of the proposed scheme (cost reduction by $49.6\%$, $36.5\%$ and $13.5\%$, respectively) is the best.

Fig. \ref{fig:CDF_and_barplot}(b) sketches the bar chart of the four components of per-frame cost.
BM1 has the highest outdated AoI penalty, because its scheduling is independent of the system state.
BM2 has the lowest AoI penalty, but consumes significantly higher transmission energy than the proposed scheme.
BM3 has the highest sampling energy consumption and thus highest sampling frequency, because it tends to select the sensor with high throughput.
Hence, the proposed scheme achieves the best balance between AoI and energy consumption.
\begin{figure}[htb]
	\centering
	\subfloat[]{
		\includegraphics[width=0.47\linewidth]{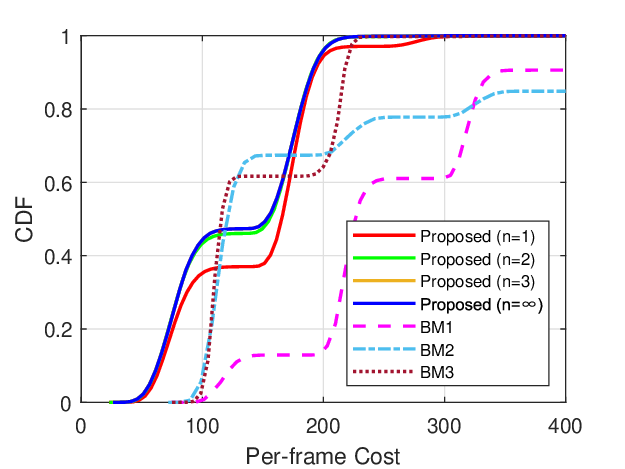}
	}
	\hfill
	\subfloat[]{
		\includegraphics[width=0.47\linewidth]{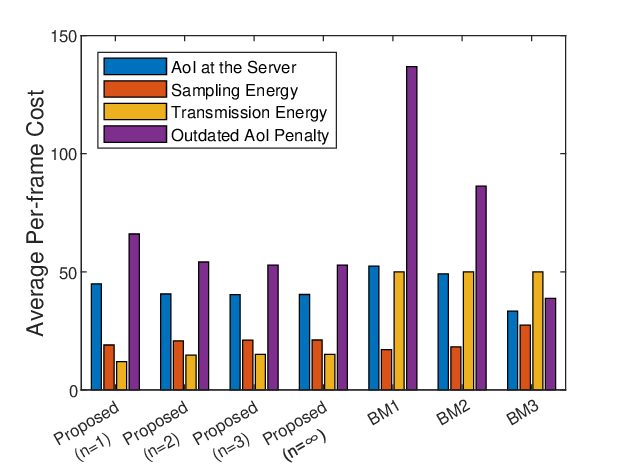}
	}
	\caption{(a) CDF of the per-frame cost.
		(b) Average performance of the AoI at the server, sampling energy, transmission energy and outdated AoI penalty.}
	\label{fig:CDF_and_barplot}
\end{figure}

The insights on blockage-predictive scheduling of the proposed policy can be obtained in Fig. \ref{fig:insight}, which shows the dynamics of the AoI at the $4$-th sensor and the corresponding AoI at the server.
The trajectory of the blocker for this trace is illustrated in Fig. \ref{fig:layout}.
The LoS path between the BS and the $4$-th sensor is blocked since the $45$-th frame.
Compared with BM1, the proposed policy can detect future channel degradation, and keep the AoIs at the $4$-th sensor and the BS at a low level before the $45$-th frame, which reduces the time duration with outdated AoI at the server.
\begin{figure}[htb]
	\centering
	\includegraphics[width=0.85\linewidth]{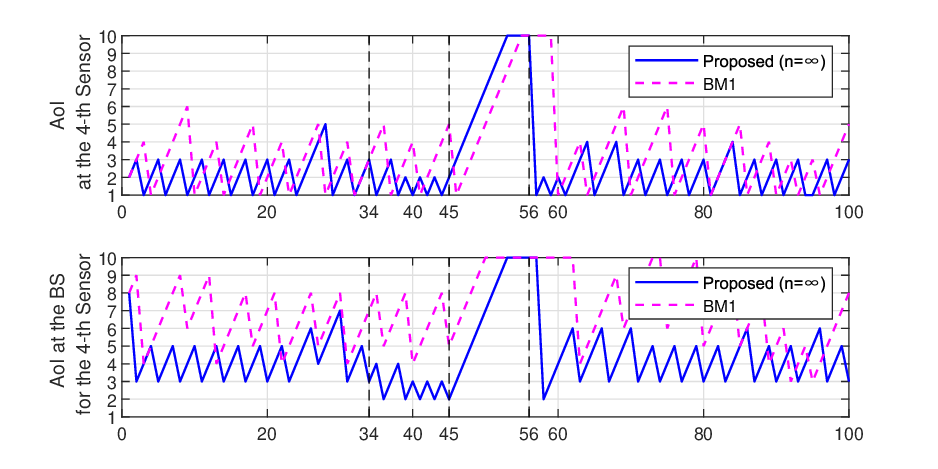}
	\caption{AoI dynamics of the $4$-th sensor.}
	\label{fig:insight}
\end{figure}

The average per-frame cost versus the number of sensors is studied in Fig. \ref{fig:SenAnaK}.
The average per-frame cost of the proposed scheme is always lower than the benchmarks, which verifies the better performance of the proposed scheme and its robustness against the number of sensors.
\begin{figure}[htb]
	\centering
	\includegraphics[width=0.6\linewidth]{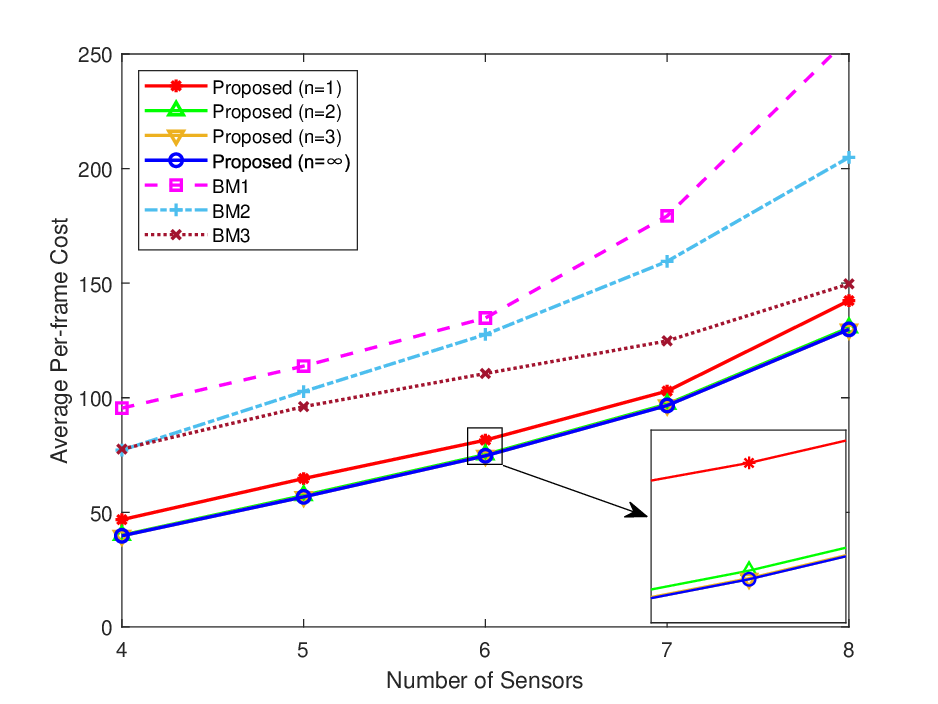}
	\caption{Average per-frame cost versus the number of sensors.}
	\label{fig:SenAnaK}
\end{figure}

\section{Conclusion}
\label{sec:conclusion}
In this paper, we formulate the dynamic scheduling of sampling and uploading in an mmWave-based WSN with random human blockage as an infinite-horizon MDP with discounted cost, where the weighted sum of average AoI and system energy consumption is the minimization objective.
Since the system state space grows exponentially with respect to the number of sensors, an approximate MDP solution framework is proposed to address the curse of dimensionality.
In the proposed framework, the optimal value function is approximated by an analytical expression derived from a reference policy.
The numerical evaluation of the value function in conventional approximate MDP solutions can be eliminated.
Finally, the policy iteration based on analytical value function is significantly more efficient than lookup-table-based value functions in conventional solution frameworks.
Thus, the solution complexity is greatly reduced.
Moreover, the impact of random human blockage on future cost is predicted in the approximate value function, and hence, mitigated in the scheduling of the current frame.
In future work, we shall consider low-complexity scheduling for multi-blocker scenarios, where how to address the curse of dimensionality caused by large location state space is challenging.

%\appendices
\section*{Appendix A: PROOF of LEMMA \ref{lem:PMF_D}}
\label{proof:PMF_D}
According to \cite{1033686}, with sufficient large $N_{\mathrm{R}}$ and $N_{\mathrm{T}}$, $f_{\mathrm{R}}(\phi)\!\rightarrow\!1$ and $f_{\mathrm{T}}(\theta)\!\rightarrow\!1$ if $\phi=0$ and $\theta=0$, otherwise $f_{\mathrm{R}}(\phi)\!\rightarrow\!0$ and $f_{\mathrm{T}}(\theta)\!\rightarrow\!0$, where $f_{\mathrm{R}}(\phi)\!=\!|\mathbf{a}_{\mathrm{R}}^{\mathsf{H}}(0)\mathbf{a}_{\mathrm{R}}(\phi)|$ and $f_{\mathrm{T}}(\theta)\!=\!|\mathbf{a}_{\mathrm{T}}^{\mathsf{H}}(0)\mathbf{a}_{\mathrm{T}}(\theta)|$.

Hence, substituting \eqref{eqn:H_tk} into \eqref{eqn:beamforming_1}, we can derive
\begin{align}
	\label{eqn:lem1_1}
	\big(\mathbf{w}_{t\!,k},\mathbf{f}_{t\!,k}\big)
	=
	 & \mathop{\arg\max}_{{}
	^{\mathbf{a}_{\mathrm{R}}(\phi_{q})\in\mathcal{W}}
	_{\mathbf{a}_{\mathrm{T}}(\phi_{p})\in\mathcal{F}}}
	\mathbb{E}_{(\alpha_{t\!,k\!,i})_{i\!\in\!\mathcal{M}}}
	\Big|
	\sum_{i\!\in\!\mathcal{M}}B_{t\!,k\!,i}(l_{t}^{\mathrm{B}})\alpha_{t\!,k\!,i} \nonumber \\
	 & \times f_{\mathrm{R}}(\phi_{k\!,i}\!-\!\phi_{q})
	f_{\mathrm{T}}(\theta_{p}\!-\!\theta_{k\!,i})
	\Big|^{2}                                                                               \\
	\label{eqn:lem1_2}
	=
	 & \mathop{\arg\max}_{{}
	{\mathbf{a}_{\mathrm{R}}(\phi_{k\!,i})},
	{\mathbf{a}_{\mathrm{T}}(\phi_{k\!,i})}}
	B_{t\!,k\!,i}(l_{t}^{\mathrm{B}})\rho_{k\!,i}^{-1}.
\end{align}
Let $i^{\star}\!=\!\mathop{\arg\max}_{i}B_{t,k,i}(l_{t}^{\mathrm{B}})\rho_{k,i}^{-1}$, \eqref{eqn:lem1_2} can be written by $
	(\mathbf{w}_{t\!,k},\mathbf{f}_{t\!,k})
	\!=\!
	(\mathbf{a}_{\mathrm{R}}(\phi_{k\!,i^{\star}}),\mathbf{a}_{\mathrm{T}}(\phi_{k\!,i^{\star}}))
$.
Then the baseband gain can be represented by
$
	Y_{t,k}(l_{t}^{\mathrm{B}})
	\!=\!
	|\alpha_{t,k,i^{\star}}|^{2}/(N_{0}W)
$.

%\begin{align}
%	Y_{t,k}(l_{t}^{\mathrm{B}})
%	\!=\!
%	\frac{\big|\alpha_{t,k,i^{\star}}\big|^{2}}{N_{0}W}.
%\end{align}

Given $
	\alpha_{t,k,i^{\star}}\!\sim\!\mathcal{CN}(0,\rho_{k,i^{\star}}^{-1})
$, $
	|\alpha_{t,k,i^{\star}}|^{2}\!=\!(\Re[\alpha_{t,k,i^{\star}}])^{2}\!+\!(\Im[\alpha_{t,k,i^{\star}}])^{2}
$ follows an exponential distribution, i.e., $|\alpha_{t,k,i^{\star}}|^{2}\!\sim\!\mathrm{Exp}(\rho_{k,i^{\star}})$.
Then $Y_{t,k}(l_{t}^{\mathrm{B}})\!\sim\!\mathrm{Exp}(\rho_{k,i^{\star}}N_{0}W)$ and thus the CDF of $Y_{t,k}(l_{t}^{\mathrm{B}})$ can be written by
\begin{align}
	\Pr[Y_{t,k}(l_{t}^{\mathrm{B}})\leq x]
	=
	1-\exp(-\rho_{k,i^{\star}}N_{0}Wx).
\end{align}
Therefore, the PMF of departure packet number is
\begin{align*}
	 & \Pr\left[D_{t,k}^{\Pi}\!=\!d\big|l_{t}^{\mathrm{B}}\right]
	\!=\!
	\Pr\left[D_{t,k}^{\Pi}\!\leq\!d\!+\!1\big|l_{t}^{\mathrm{B}}\right]
	\!-\!
	\Pr\left[D_{t,k}^{\Pi}\!\leq\!d\big|l_{t}^{\mathrm{B}}\right] \nonumber                                                          \\
	\!=\!
	 & \Pr\bigg[Y_{t,k}(l_{t}^{\mathrm{B}})\!\leq\!\frac{2^{\frac{(d\!+\!1)N_{\mathrm{b}}}{W\tau_{t,k}^{\Pi}}}\!-\!1}{P^{\Pi}}\bigg]
	\!-\!\Pr\bigg[Y_{t,k}(l_{t}^{\mathrm{B}})\!\leq\!\frac{2^{\frac{dN_{\mathrm{b}}}{W\tau_{t,k}^{\Pi}}}\!-\!1}{P^{\Pi}}\bigg].
\end{align*}

\section*{Appendix B: PROOF of LEMMA \ref{lem:P_k}}
\label{proof:P_k}
In fact, $\mathbf{M}_{k}^{(\ell)}$ represents the transition probability matrix of local abstract state eliminating $l_{t}^{\mathrm{B}}$, i.e., $
	(Q_{t,k},A_{t,k}^{\mathrm{s}},A_{t,k}^{\mathrm{d}})
$, conditioned on $l_{t+1}^{\mathrm{B}}\!=\!\ell$.
Then the derivation of \eqref{eqn:P_k} is straightforward.
We have the following discussion on all possible cases for $\mathbf{M}_{k}^{(\ell)}$ defined in Table \ref{tab:M_kl}.
\begin{itemize}[leftmargin=10pt]
	\item
	      \textbf{Case 1 ($
			      Q_{t\!,k}{=}0,
			      1{\leq} A_{t\!,k}^{\mathrm{s}}{\leq} A_{\mathrm{max}},
			      1{\leq} A_{t\!,k}^{\mathrm{d}}{\leq} A_{\mathrm{max}},
			      Q_{t\!+\!1\!,k}{=}0,\allowbreak
			      A_{t\!+\!1\!,k}^{\mathrm{s}}{=}1,
			      A_{t\!+\!1\!,k}^{\mathrm{d}}{=}1
		      $):}
	      This means that sampling and transmission of all packets of the new sample are accomplished with the probability of $L_{k}$ departure packets.

	\item
	      \textbf{Case 2 ($Q_{t\!,k}{=}0,
		      1{\leq} A_{t\!,k}^{\mathrm{s}}{\leq} A_{\mathrm{max}},
		      1{\leq} A_{t\!,k}^{\mathrm{d}}{\leq} A_{\mathrm{max}},
		      1{\leq} Q_{t\!+\!1\!,k}\allowbreak {\leq} L_{k},
		      A_{t\!+\!1\!,k}^{\mathrm{s}}{=}1,
		      A_{t\!+\!1\!,k}^{\mathrm{d}}{=}\min\{A_{t\!,k}^{\mathrm{d}}{+}1,A_{\mathrm{max}}\}
	      $):}
	      This means that sampling and transmission of $(L_{k}\!-\!Q_{t\!+\!1\!,k})$ packets of the new sample are accomplished with the probability of $(L_{k}\!-\!Q_{t\!+\!1\!,k})$ departure packets.

	\item
	      \textbf{Case 3 ($
		      1{\leq} Q_{t\!,k}{\leq} L_{k},
		      1{\leq} A_{t\!,k}^{\mathrm{s}}{\leq} A_{\mathrm{max}},
		      1{\leq} A_{t\!,k}^{\mathrm{d}}{\leq} A_{\mathrm{max}},
		      Q_{t\!+\!1\!,k}\allowbreak {=}0,
		      A_{t\!+\!1\!,k}^{\mathrm{s}}{=}\min\{A_{t\!,k}^{\mathrm{s}}{+}1,A_{\mathrm{max}}\},
		      A_{t\!+\!1\!,k}^{\mathrm{d}}{=}\min\{A_{t\!,k}^{\mathrm{s}}{+}1,\allowbreak A_{\mathrm{max}}\}
	      $):}
	      This means that transmission of all remaining $Q_{t\!+\!1\!,k}$ packets of the current sample is accomplished with the probability of $Q_{t\!+\!1\!,k}$ departure packets.

	\item
	      \textbf{Case 4 ($
		      1{\leq} Q_{t\!,k}{\leq} L_{k},
		      1{\leq} A_{t\!,k}^{\mathrm{s}}{\leq} A_{\mathrm{max}},
		      1{\leq} A_{t\!,k}^{\mathrm{d}}{\leq} A_{\mathrm{max}},
		      1{\leq} \allowbreak Q_{t\!+\!1\!,k}{\leq} Q_{t\!,k},
		      A_{t\!+\!1\!,k}^{\mathrm{s}}{=}\min\{A_{t\!,k}^{\mathrm{s}}{+}1,A_{\mathrm{max}}\},
		      A_{t\!+\!1\!,k}^{\mathrm{d}}{=}\min\allowbreak \{A_{t\!,k}^{\mathrm{d}}{+}1,A_{\mathrm{max}}\}
	      $):}
	      This means that transmission of $(Q_{t\!,k}{-}Q_{t\!+\!1\!,k})$ packets of the current sample is accomplished with the probability of $(Q_{t\!,k}{-}Q_{t\!+\!1\!,k})$ departure packets.
\end{itemize}

\section*{Appendix C: PROOF of THEOREM \ref{thm:analytical_expression}}
\label{proof:analytical_expression}
The approximate value function is given by
\begin{align}
	\label{eqn:proof_the1_1}
	W^{\Pi}(\widetilde{\mathcal{S}}_{t})
	\!=\!
	 & \lim_{T\rightarrow\infty}
	\sum_{t=1}^{T}
	\gamma^{t\!-\!1}
	\Big[\sum_{k\in\mathcal{K}}\big[\!
	\mathbf{e}_{\kappa(l_{t}^{\mathrm{B}}\!,\epsilon(Q_{t\!,k},A_{t\!,k}^{\mathrm{s}},A_{t\!,k}^{\mathrm{d}}))}^{|\mathbb{L}|(L_{k}+1)A_{\mathrm{max}}^{2}}\big]^{\!\mathsf{T}}
	\mathbf{P}_{k}^{t\!-\!1}
	\mathbf{g}_{k}
	\nonumber
	\\
	 & +w_{\mathrm{P}}T_{\mathrm{F}}P^{\Pi}\Big].
\end{align}
Since the reference policy adopts constant transmission time and power allocations, the per-frame cost for transmission power consumption, and thus its discounted cumulative sum corresponding to the second term in \eqref{eqn:local_value_function}, results into a constant.
In fact, $
	\mathbf{e}_{\kappa(l_{t}^{\mathrm{B}},\epsilon(Q_{t\!,k},A_{t\!,k}^{\mathrm{s}},A_{t\!,k}^{\mathrm{d}}))}^{|\mathbb{L}|(L_{k}+1)A_{\mathrm{max}}^{2}}
$ represents the probability vector for a deterministic local abstract state $\widetilde{\mathcal{S}}_{t,k}\!=\!(l_{t}^{\mathrm{B}},Q_{t\!,k},A_{t\!,k}^{\mathrm{s}},A_{t\!,k}^{\mathrm{d}})$, $\mathbf{P}_{k}$ represents the transition probability matrix, and the $\kappa(l_{t}^{\mathrm{B}},\epsilon(Q_{t\!,k},A_{t\!,k}^{\mathrm{s}},A_{t\!,k}^{\mathrm{d}}))$-th entry of $\mathbf{g}_{k}$ represents the per-frame cost for sampling energy cost, and AoIs and outdated AoI penalties at the server.
Since the reference policy samples only when the queue becomes empty, the per-frame cost for sampling is counted only for $Q_{t,k}\!=\!0$.
Therefore, $
	\big[\mathbf{e}_{\kappa(l_{t}^{\mathrm{B}},\epsilon(Q_{t\!,k},A_{t\!,k}^{\mathrm{s}},A_{t\!,k}^{\mathrm{d}}))}^{|\mathbb{L}|(L_{k}+1)A_{\mathrm{max}}^{2}}\big]^{\mathsf{T}}\mathbf{P}_{k}^{t-1}\mathbf{g}_{k}
$ represents the expected per-frame cost in the $t$-th frame.
The derivation from \eqref{eqn:proof_the1_1} to \eqref{eqn:local_value_function} resembles the proof in Appendix B(3) of \cite{9078843}.

\bibliography{IEEEabrv,reference_ICPADS23_v2_230830_abrv}

\begin{thebibliography}{10}
\providecommand{\url}[1]{#1}
\csname url@samestyle\endcsname
\providecommand{\newblock}{\relax}
\providecommand{\bibinfo}[2]{#2}
\providecommand{\BIBentrySTDinterwordspacing}{\spaceskip=0pt\relax}
\providecommand{\BIBentryALTinterwordstretchfactor}{4}
\providecommand{\BIBentryALTinterwordspacing}{\spaceskip=\fontdimen2\font plus
\BIBentryALTinterwordstretchfactor\fontdimen3\font minus
  \fontdimen4\font\relax}
\providecommand{\BIBforeignlanguage}[2]{{%
\expandafter\ifx\csname l@#1\endcsname\relax
\typeout{** WARNING: IEEEtran.bst: No hyphenation pattern has been}%
\typeout{** loaded for the language `#1'. Using the pattern for}%
\typeout{** the default language instead.}%
\else
\language=\csname l@#1\endcsname
\fi
#2}}
\providecommand{\BIBdecl}{\relax}
\BIBdecl

\bibitem{6515173}
T.~S. Rappaport, S.~Sun, R.~Mayzus, H.~Zhao, Y.~Azar, K.~Wang, G.~N. Wong,
  J.~K. Schulz, M.~Samimi, and F.~Gutierrez, ``Millimeter wave mobile
  communications for {5G} cellular: It will work!'' \emph{{IEEE} Access},
  vol.~1, pp. 335--349, 2013.

\bibitem{Kenney-2011}
S.~Kaul, M.~Gruteser, V.~Rai, and J.~Kenney, ``Minimizing age of information in
  vehicular networks,'' in \emph{2011 8th Annual IEEE Communications Society
  Conference on Sensor, Mesh and Ad Hoc Communications and Networks}, June
  2011, pp. 350--358.

\bibitem{8123937}
X.~Wu, J.~Yang, and J.~Wu, ``Optimal status update for age of information
  minimization with an energy harvesting source,'' \emph{IEEE Transactions on
  Green Communications and Networking}, vol.~2, no.~1, pp. 193--204, 2018.

\bibitem{9729335}
M.~Sheikhi and V.~Hakami, ``{AoI}-aware status update control for an energy
  harvesting source over an uplink {mmWave} channel,'' in \emph{2021 7th
  International Conference on Signal Processing and Intelligent Systems
  (ICSPIS)}, 2021, pp. 01--06.

\bibitem{8778671}
B.~Zhou and W.~Saad, ``Joint status sampling and updating for minimizing age of
  information in the internet of things,'' \emph{{IEEE} Trans. Commun.},
  vol.~67, no.~11, pp. 7468--7482, 2019.

\bibitem{8938128}
------, ``Minimum age of information in the internet of things with non-uniform
  status packet sizes,'' \emph{{IEEE} Trans. Wireless Commun.}, vol.~19, no.~3,
  pp. 1933--1947, 2020.

\bibitem{8736523}
X.~Zheng, S.~Zhou, Z.~Jiang, and Z.~Niu, ``Closed-form analysis of non-linear
  age of information in status updates with an energy harvesting transmitter,''
  \emph{{IEEE} Trans. Wireless Commun.}, vol.~18, no.~8, pp. 4129--4142, 2019.

\bibitem{9796654}
X.~He, S.~Wang, X.~Wang, S.~Xu, and J.~Ren, ``Age-based scheduling for
  monitoring and control applications in mobile edge computing systems,'' in
  \emph{IEEE INFOCOM 2022 - IEEE Conference on Computer Communications}, 2022,
  pp. 1009--1018.

\bibitem{BoZhou-2018-GC}
B.~Zhou and W.~Saad, ``Optimal sampling and updating for minimizing age of
  information in the internet of things,'' in \emph{2018 IEEE Global
  Communications Conference (GLOBECOM)}, Dec 2018, pp. 1--6.

\bibitem{JADE}
J.~Palacios, P.~Casari, and J.~Widmer, ``{JADE}: Zero-knowledge device
  localization and environment mapping for millimeter wave systems,'' in
  \emph{IEEE INFOCOM 2017 - IEEE Conference on Computer Communications}, 2017,
  pp. 1--9.

\bibitem{E-Mi}
T.~Wei, A.~Zhou, and X.~Zhang, ``Facilitating robust 60 {GHz} network
  deployment by sensing ambient reflectors,'' in \emph{14th USENIX Symposium on
  Networked Systems Design and Implementation (NSDI 17)}, 2017, pp. 213--226.

\bibitem{10005216}
Y.~Sun, J.~Li, T.~Zhang, R.~Wang, X.~Peng, X.~Han, and H.~Tan, ``An indoor
  environment sensing and localization system via mmwave phased array,''
  \emph{Journal of Communications and Information Networks}, vol.~7, no.~4, pp.
  383--393, 2022.

\bibitem{9593198}
G.~Li, S.~Wang, J.~Li, R.~Wang, X.~Peng, and T.~X. Han, ``Wireless sensing with
  deep spectrogram network and primitive based autoregressive hybrid channel
  model,'' in \emph{2021 IEEE 22nd International Workshop on Signal Processing
  Advances in Wireless Communications (SPAWC)}, 2021, pp. 481--485.

\bibitem{9737357}
F.~Liu, Y.~Cui, C.~Masouros, J.~Xu, T.~X. Han, Y.~C. Eldar, and S.~Buzzi,
  ``Integrated sensing and communications: Toward dual-functional wireless
  networks for {6G} and beyond,'' \emph{{IEEE} J. Sel. Areas Commun.}, vol.~40,
  no.~6, pp. 1728--1767, 2022.

\bibitem{9198891}
R.~Zhang, X.~Jing, S.~Wu, C.~Jiang, J.~Mu, and F.~R. Yu, ``Device-free wireless
  sensing for human detection: The deep learning perspective,'' \emph{{IEEE}
  Internet Things J.}, vol.~8, no.~4, pp. 2517--2539, 2021.

\bibitem{6717211}
O.~E. Ayach, S.~Rajagopal, S.~Abu-Surra, Z.~Pi, and R.~W. Heath, ``Spatially
  sparse precoding in millimeter wave {MIMO} systems,'' \emph{{IEEE} Trans.
  Wireless Commun.}, vol.~13, no.~3, pp. 1499--1513, 2014.

\bibitem{8764465}
Y.~Sun and B.~Cyr, ``Sampling for data freshness optimization: Non-linear age
  functions,'' \emph{Journal of Communications and Networks}, vol.~21, no.~3,
  pp. 204--219, 2019.

\bibitem{bertsekas2012dynamic}
D.~Bertsekas, \emph{Dynamic programming and optimal control}, 4th~ed.\hskip 1em
  plus 0.5em minus 0.4em\relax Belmont, MA, USA: Athena scientific, 2012,
  vol.~2.

\bibitem{li2006towards}
L.~Li, T.~J. Walsh, and M.~L. Littman, ``Towards a unified theory of state
  abstraction for {MDPs}.'' in \emph{AI\&M}, 2006.

\bibitem{9145118}
Q.~Zeng, Y.~Du, K.~Huang, and K.~K. Leung, ``Energy-efficient radio resource
  allocation for federated edge learning,'' in \emph{2020 IEEE International
  Conference on Communications Workshops (ICC Workshops)}, 2020, pp. 1--6.

\bibitem{maltsev2010channel}
A.~Maltsev, V.~Erceg, E.~Perahia, C.~Hansen, R.~Maslennikov, A.~Lomayev,
  A.~Sevastyanov, A.~Khoryaev, G.~Morozov, M.~Jacob \emph{et~al.}, ``Channel
  models for 60 {GHz} {WLAN} systems, doc.: Ieee 802.11-09/0334r8,'' \emph{IEEE
  802.11 document 09/0334r8}, 2010.

\bibitem{georgiadis2006resource}
L.~Georgiadis, M.~J. Neely, L.~Tassiulas \emph{et~al.}, ``Resource allocation
  and cross-layer control in wireless networks,'' \emph{Foundations and
  Trends{\textregistered} in Networking}, vol.~1, no.~1, pp. 1--144, 2006.

\bibitem{1033686}
A.~Sayeed, ``Deconstructing multiantenna fading channels,'' \emph{{IEEE} Trans.
  Signal Process.}, vol.~50, no.~10, pp. 2563--2579, 2002.

\bibitem{9078843}
S.~Huang, B.~Lv, R.~Wang, and K.~Huang, ``Scheduling for mobile edge computing
  with random user arrivals—an approximate {MDP} and reinforcement learning
  approach,'' \emph{{IEEE} Trans. Veh. Technol.}, vol.~69, no.~7, pp.
  7735--7750, 2020.

\end{thebibliography}
\bibliographystyle{IEEEtran}

\end{document}